\documentclass[12pt]{article}
\usepackage[utf8]{inputenc}

\usepackage{natbib}
\setcitestyle{authoryear,open={(},close={)}}

\usepackage{textcomp,gensymb} 

\usepackage{nicefrac}
\usepackage[normalem]{ulem}

\usepackage[margin=1.2in]{geometry}

\usepackage{hyperref}
\hypersetup{colorlinks=true,linkcolor=blue,citecolor=blue}
\usepackage{dsfont} 

\usepackage{lmodern, microtype}
\usepackage[small]{titlesec}

\usepackage{comment} 
\usepackage{graphics,graphicx}
\usepackage{tikz}
\usepackage{pgfplots}          
\pgfplotsset{compat=1.18}     
\usetikzlibrary{plotmarks} 
\usepackage{float,caption,subcaption}
\usepackage{amsmath,mathtools,amsthm,amssymb,amsfonts}
\usepackage{enumitem}

\newcommand{\diag}{\mathrm{diag}}

\newcommand*\diff{\mathop{}\!\mathrm{d}}

\def\ep{\varepsilon}

\def\al{\alpha}

\renewcommand{\P}{\mathbb{P}}
\newcommand{\R}{\mathbb{R}}

\newcommand{\supp}{\mathrm{supp}}

\definecolor{ForestGreen}{rgb}{.13,.54,.13}
\definecolor{BrickRed}{rgb}{.80,.26,.33}

\ifdefined\DRAFT

\newcommand{\fed}[1]{\textcolor{BrickRed}{(\textbf{Fedor:} #1)}}
\newcommand{\fede}[1]{\textcolor{ForestGreen}{(\textbf{FE:} #1)}}
\newcommand{\joe}[1]{\textcolor{cyan}{(\textbf{Joe:} #1)}}

\else
\newcommand{\fed}[1]{}
\newcommand{\fede}[1]{}
\newcommand{\joe}[1]{}

\fi

\theoremstyle{plain}
\newtheorem{theorem}{Theorem}
\newtheorem{lemma}{Lemma}

\newtheorem{corollary}{Corollary}
\newtheorem{proposition}{Proposition}

\theoremstyle{definition}
\newtheorem{definition}{Definition}
\theoremstyle{remark}

\newtheorem*{remark*}{Remark}
\newtheorem{example}{Example}

\title{Stable Matching as Transport:\\
{\large a Welfarist Perspective on Market Design}\thanks{
The paper has benefited from suggestions and comments by Dmitry Arkhangelsky, Sophie Bade, Peter Biro, Mark Braverman, Ben Brooks, Modibo Camara, Rodrigo Crousillat Rayter, Ryota Iijima, David Kempe, Axel Niemeyer, Afshin Nikzad, Piyush Panigrahi, Marek Pycia, Doron Ravid, Phil Reny, Kirill Rudov, Utku Unver, Rakesh Vohra, Matt Weinberg, Leeat Yariv, and audiences at ASU, U.\ of Arizona, Berkeley, Caltech, CMU, Columbia, Emory, Harvard, NYU, MIT, MSRI, Princeton, Simons Institute, UPenn, Washington University, and Yale.
}}
\author{\normalsize Federico Echenique\thanks{Department of Economics, University of California, Berkeley} \and \normalsize Joseph Root\thanks{Department of Economics, University of Chicago} \and \normalsize Fedor Sandomirskiy\thanks{Department of Economics, Princeton University}}

\ifdefined\LINKTOLATESTVERSION
\date{\normalsize \today \\ the latest  version: \href{https://fedors.info/papers/2024stabletransport/matching_as_transport.pdf}{\bf link}
}
\else
\date{}
\fi

\begin{document}

\maketitle
\begin{abstract}
This paper links matching markets with aligned preferences to optimal transport theory. We show that stability, efficiency, and fairness emerge as solutions to a parametric family of optimal transport problems. The parameter indexes a planner's attitude towards inequality. This link offers insights into structural properties of matchings and trade-offs between objectives, showing how stability can lead to welfare inequalities, even among similar agents. Our model captures supply-demand imbalances in contexts like spatial markets, school choice, and ride-sharing. We also show that large markets with idiosyncratic preferences can be well approximated by aligned preferences, expanding the applicability of our results.
\phantom{aaaaaa}\\
\vskip 0.2cm
\phantom{aaaaaa}
\end{abstract}

\newpage
\ifdefined\DRAFT

\thispagestyle{empty}
\tableofcontents

\clearpage
\setcounter{page}{1}

\fi

\section{Introduction}
Stability has become the central organizing principle in the design of matching markets. It is usually defended as a fairness requirement and imposed as a game-theoretic equilibrium constraint on feasible matchings. Framed in this way, stability appears fundamentally different from the explicit objective-maximizing criteria that dominate mechanism design. Our paper offers an alternative perspective. We show that stable matchings can be obtained as limits of solutions to well-defined optimization problems arising in the theory of optimal transport. This reveals an implicit objective optimized by stable matchings, and unifies stability, welfare, and equity in a single framework. Once expressed in this common language, the objectives that these criteria seek to maximize can be directly compared.

We prove that, in a wide range of economically relevant matching markets, the stable matching maximizes a particular social welfare function. Unlike standard utilitarian welfare, this objective places greater weight on high-utility matches and lower weight on low-utility ones. In other words, the notion of social welfare built into stability is far more tolerant of inequality than the usual utilitarian benchmark.

We begin with markets in which preferences are aligned, meaning that a potential function $U$ provides a common index of match quality for both sides of the market. To fix ideas, consider a school district where students prefer nearby schools (say, because school quality is roughly uniform) and schools assign priority to nearby students (a neighborhood-priority rule). Both sides’ preferences are then generated by a common match-quality function
\[
U(x,y) = -\mathrm{distance}(x,y),
\]
so each student $x$ ranks schools by $U(x\,,\,\cdot)$ and each school $y$ ranks students by $U(\,\cdot\,,y)$.

In aligned markets, we show that the stable matching optimizes a version of \emph{Atkinson's inequality index} (Theorem~\ref{th_equivalence}). This index aggregates transformed match qualities $\varphi_{\alpha}(U(x,y))$, where the parameter $\alpha \in \mathbb{R}$ governs the implicit attitude toward inequality. Varying $\alpha$ traces a continuous frontier between canonical benchmarks: $\alpha = 0$ yields utilitarian welfare maximization; $\alpha \to -\infty$ approaches an egalitarian criterion that favors the lower tail of the welfare distribution; and $\alpha \to +\infty$ favors inequality by increasing the weight on high-utility matches at the expense of low-utility ones. We show that the latter case corresponds precisely to stable matchings.

Our results redirect attention to a form of inequality that the matching literature has largely overlooked. Instead of the classic conflict of interest between the two sides of the market (student- or worker-optimal versus school- or firm-optimal), we expose the inequality that stability itself creates \emph{within} each side. By design, stability favors high-utility matches and ignores the negative externalities that those matches, through supply-and-demand imbalances, impose on the rest of the market. The resulting within-side inequality can be stark: tiny differences in types or locations can translate into large welfare gaps, even though more equitable outcomes are feasible. 

At the same time, the use of Atkinson's index with a large positive $\alpha$ does not render stability arbitrarily inefficient or unfair. In aligned environments, we prove that a stable matching achieves at least one-half of the optimal utilitarian welfare and attains an egalitarian outcome at least one-half as good as the optimal egalitarian benchmark (Theorem~\ref{th_welfare}). These bounds provide worst-case performance guarantees for stability under alignment.

The results we have described so far apply to markets with aligned preferences, which may seem like a restrictive assumption. In practice, preferences depend on many dimensions: for example, school quality, student achievement, subject specialization, and sibling priorities. These need not line up across the two sides of the market. Nevertheless, we demonstrate that alignment is more relevant, and our results more broadly applicable, than they might at first appear. 

Many of the dimensions that drive preferences represent vertical quality differentiation and, as we discuss in Section~\ref{sec:motivation}, can be absorbed into the potential $U$. Still, some forms of heterogeneity cannot be captured by a common potential. For example, horizontal preferences, such as variation in individual taste, can lead to misalignment. We therefore extend our analysis to markets with both a common aligned component and horizontal taste shocks.  We show that, if the market is sufficiently large, every matching lies near another matching in which each agent receives an almost perfect partner with respect to the idiosyncratic component  (Theorem~\ref{th_idiosyncractic}). In this sense, large markets are well approximated by an aligned benchmark and the allocation over the aligned component can serve as a useful benchmark for comparing stability, fairness, and welfare in sufficiently large markets. Importantly, the approximation error admits an explicit bound applicable to finite markets.

To leverage the power of optimal transport to understand stable matchings, we further study spatial markets where agents are embedded in $\R^d$, and match quality is determined by distance. The optimal transport formulation yields novel structural insights into stable matchings, with particularly clear results in the one-dimensional case. We identify when there is a conflict between stability, welfare, and fairness, yielding clear benchmarks for real-world settings like school choice and ride-sharing.

Together, our results recast stability as embodying concrete (if likely unintended) policy objectives. In doing so, the paper contributes to a more transparent account of what stability delivers and what it sacrifices, and it provides a tractable benchmark model for analyzing the trade-offs at stake.

\paragraph{Related literature.}
The economic literature on matching markets with aligned preferences remains sparse. A systematic study was initiated by \cite*{niederle2009decentralized} and \cite*{ferdowsian2020decentralized}, who examine the convergence of decentralized matching dynamics in such markets.
\cite*{peski2011smooth} considers a continuous market with aligned preferences induced by $U(x,y)=\sum_{i=1}^d (x_i+y_i)$ for $x,y\in\R^d$, so that all agents prefer higher-coordinate matches; he shows that stable matchings exhibit a rich structure that can be characterized via solutions to partial differential equations. \cite{Clark2020} considers a matching model with one-dimensional types, transferable utility, and \(U(x,y)=f(|x-y|)\). He focuses on patterns of stable matching, and in particular on when matchings are positively or negatively assortative. \cite{Clark2025} considers these issues in a model with imperfectly transferable utility, and \cite{flanders2018} under single-peaked preferences. 
For finite populations with aligned preferences, stable matchings can be obtained via a greedy algorithm that iteratively matches the unmatched pair with the highest utility; this observation is implicit in \cite*{eeckhout2000uniqueness} and \cite*{clark2006uniqueness}, who derive conditions for the uniqueness of stable matchings; see also \cite*{romero2013acyclicity, reny2021simple, gutin2023unique}. \cite*{galichon2023stable} extend the greedy algorithm to many-to-one matching and document its unfairness in simulations, a phenomenon we explore analytically. 

Our paper contributes to the understanding of the stability-efficiency trade-off in markets with aligned preferences, complementing insights by \cite*{cantillon2022respecting} and \cite*{lee2014efficiency}. \cite*{cantillon2022respecting} 
introduce a school choice environment where preference alignment becomes a descriptive property. Under a condition generalizing alignment, 
they demonstrate that stable matchings are Pareto optimal. By contrast, our work adopts a quantitative approach to efficiency, demonstrating that Pareto optimality may entail a loss of utilitarian welfare. \cite*{lee2014efficiency} study large markets with random preferences, showing that when a common utility for each pair is drawn independently from a continuous distribution, the stable matching achieves utilitarian welfare close to the optimum. Our results reveal that this phenomenon is specific to settings with no correlation across agents' preferences: in the presence of correlation, stability may reduce welfare by up to half of the optimal level.

A broader perspective on preference alignment is provided by \cite*{pycia2012stability}, who examines a general coalition-formation model that includes many-to-one matching as a special case. His work focuses on allowing for complementarities and peer effects, proving the existence of stable outcomes under a richness assumption on feasible coalitions and preferences. \cite*{echenique2007solution} make a related point for many-to-one matching, building on the earlier ideas of \cite*{banerjee2001core} in coalition formation. Pycia's work also precedes our discussion in Section~\ref{sec:applications} on the connection between preference alignment and second-stage bargaining. 

Our study touches on several other strains of literature. In computer science, aligned preferences have appeared under the name of ``globally ranked pairs.'' \cite*{abraham2008stable} analyze 
aligned preferences in the roommates problem and \cite*{lebedev2007using} apply them to peer-to-peer networks. For an empirical use of markets with aligned preferences, see, for example, \cite*{sorensen2007smart}, \cite*{agarwal2015empirical}, and \cite{dur2022deduction}.

The connection to optimal transport established in our paper reinforces the common wisdom that matching models with a continuum of agents can be more intuitive than their finite-population counterparts; see, for example, \cite*{ashlagi2016optimal, azevedo2016supply,leshno2021cutoff,arnosti2022continuum}. 
With optimal transport methods, the cardinality of the space of agents becomes largely irrelevant for the problems we address.
Connections to optimal transport have recently been discovered and used in various areas of economic theory, including matching markets with transferable utility \citep*{gretsky1992nonatomic, galichon2022cupid}, mechanism design~\citep*{daskalakis2015strong,perez2023fraud,kolesnikov2022beckmann,mccann2023duality}, information design~\citep*{malamud2021optimal,arieli2023persuasion}, and many others; see surveys by \cite*{ekeland2010notes,carlier2012optimal,galichon2018optimal}. We find a link between stability and a particular area of optimal transport known as concave transport, pioneered by
\cite*{mccann1999exact,gangbo1996geometry}. 
To our knowledge, concave transport has not appeared in economic applications, with the exception of~\cite*{boerma2023composite}, who use it to study labor markets with transferable utility. Their work is the closest to ours in the technical dimension. Although the theory of markets with and without transfers has very few commonalities, our shared reliance on the concave transport perspective highlights its generality and importance for understanding both settings.

\section{Motivation}\label{sec:motivation}

To motivate our results, we discuss a school choice problem. Consider two populations of agents: schools and students. A generic student is referred to by $i$, and a generic school by $j$. Each agent has a fixed \textbf{type}; for a student $i$, the type is $x_i \in X$, and for a school $j$, it is $y_j \in Y$.

When a student $i$ is matched with a school $j$, the utility for $i$ is given by the sum of three components:
$$
u_{i}(j) = \underbrace{q(x_i, y_j)}_{\text{match quality}} + \underbrace{g(y_j)}_{\text{vertical quality}} + \underbrace{\xi_{i}(j).}_{\text{idiosyncratic shock}}
$$
The first term,  $q(x_i, y_j)$, represents some objective notion of match quality. An important example of objective quality is distance: if the type of each agent contains their geographical location, then $q(x_i,y_j)$ could be a decreasing function of the distance between $x_i$ and $y_j$. The role of distance in school choice has been well-documented empirically \citep{walters2018demand,laverde2022distance}. The second term,  $g(y_j)$, captures the vertical quality associated with school type $y_j$. All students agree on this measure, which could, for example, encode the past performance of students at the school. Finally, $\xi_{i}(j)$ is an idiosyncratic shock specific to the individual pair $(i, j)$. The additive specification for $u_{i}(j)$ resembles the assumptions that are common in empirical work (for example, \cite{angrist2025putting}). Similarly, the priority that school $j$ assigns to student $i$ is:
$$
v_{j}(i) = q(x_i, y_j) + h(x_i) + \eta_{j}(i),
$$
where $h(x_i)$ represents the vertical quality associated with student type $x_i$. It could, for example, reflect students' test scores or their previous educational performance. The last term,  $\eta_{j}(i)$, is an idiosyncratic shock specific to the pair $(i,j)$, but from $j$'s perspective. The idiosyncratic shocks $\eta_{j}(i)$ and $\xi_{i}(j)$ are independent.

In school choice, the deferred acceptance mechanism (DA) is commonly used to match students and schools, as it yields the student-optimal stable matching. As an alternative, one could instead choose the matching that maximizes the sum of all agents' utilities: a utilitarian criterion. We compare the DA and the utilitarian matching by simulating markets that satisfy our assumptions. The results are shown in Figure~\ref{fig:simulations_welfare}, which depicts the empirical cumulative distribution (CDF) of the agents' utilities for different market sizes ($n$ denotes the number of agents on each side of the market).

\begin{figure}[H]
    \centering
\includegraphics[width=15cm]{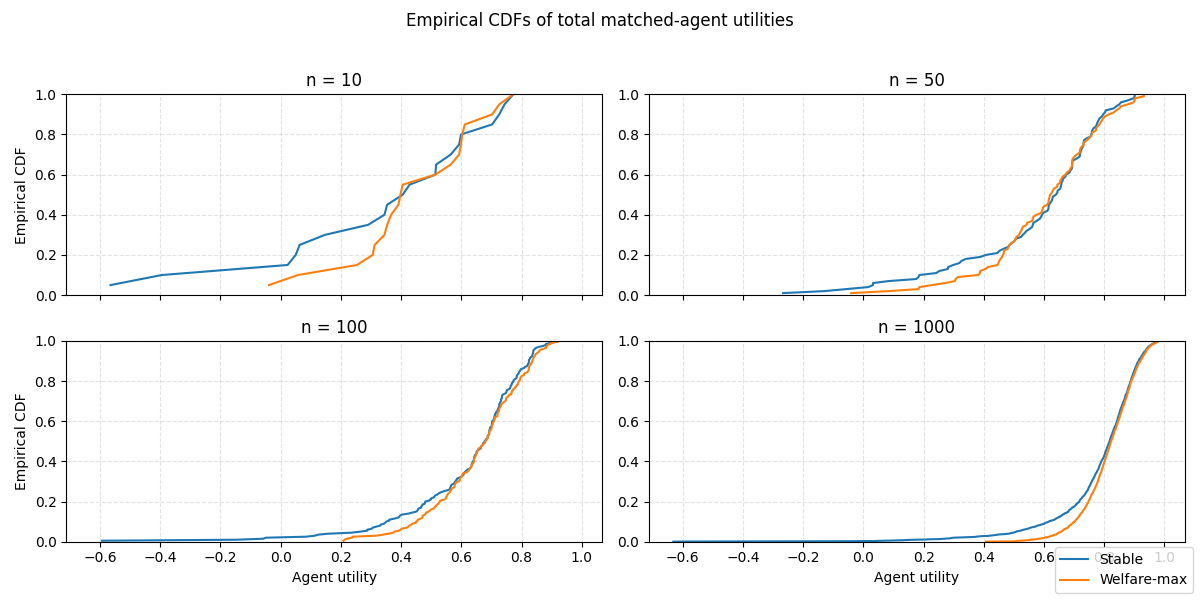}
    \caption{CDFs of matched-agent utilities under two matching mechanisms: DA and utilitarian welfare maximization. Agents on both sides are placed uniformly at random on $[0,1]^2$. Match quality is given by the negative Euclidean distance, and idiosyncratic shocks are i.i.d.\ uniform on $[0,1]$, with no vertical component. For high-utility couples, DA delivers a welfare distribution close to that under welfare maximization, but it consistently generates a long tail of low utilities across markets of all sizes.}
    \label{fig:simulations_welfare}
\end{figure}

Figure~\ref{fig:simulations_welfare} shows that, while the distributions of utility from both methods are similar, the DA generates significantly more unequal outcomes than utilitarian maximization does. The DA distribution features a long left tail of low-utility matches. This pattern is striking and perhaps surprising because utilitarian welfare maximization is neutral with respect to inequality, while the DA is often motivated by fairness.

\smallskip
Our paper helps to rationalize the observations in Figure~\ref{fig:simulations_welfare} and provides a tractable framework, using optimal transport, to model the underlying forces. The tractability of our framework stems from two insights: under certain conditions, \emph{the vertical quality terms can be absorbed into the aligned potential, while idiosyncratic shock terms can be approximated away in large markets.}

The possibility of ignoring the idiosyncratic components is a crucial step in our analysis. When a market is large, we can ensure that each agent effectively obtains the best available match in the idiosyncratic dimension, making that dimension essentially irrelevant. The result is stated in Section~\ref{sec_approx_aligned}, but the idea in a special case is easy to convey. Think of a case when $X$ and $Y$ are finite and small (two conditions that are not needed but help convey the intuition). In fact, suppose that there are only two types on each side of the market.  When the market is large, there are many agents of each type. We may then ensure for each agent $j$ that there will be an ``ideal'' partner of each type; simply because, with high probability, there will be a draw of a partner  with high idiosyncratic utility. Not only that, we may find, among those ideal partners, one agent $i$ who $j$ regards as ideal---meaning one who also derives a high idiosyncratic utility from matching with $j$. In a large market, there is a high probability that we can match agents ``ideally'' along their idiosyncratic utilities. Ideally means that the agents obtain almost all the utility that they can achieve from the idiosyncratic components of their utilities. Thus, the idiosyncratic components have little effect on the welfare distribution, and the variation that remains relevant for welfare analysis arises solely from the other components. 

Once the idiosyncratic components are removed ($\xi_{i}(j)=\eta_{j}(i)=0$ for all $i,j$), we obtain a second simplification of the analysis. Preferences of agents \emph{on both sides of the market} can be represented by a \textbf{potential}
$$
U(x_i, y_j) = q(x_i, y_j) + g(y_j) + h(x_i),
$$
so that $u_{i}(j) + h(x_i)=U(x_i, y_j)$ and $v_{j}(i)+g(y_j)=U(x_i, y_j)$. For each agent $i$, the transformed utility $u_{i}(j) + h(x_i)$ represents the same preferences as $u_{i}(j)$, because $h(x_i)$ is constant across all $j$. Similarly for $v_{j}(i)+g(y_j)$ and $v_{j}(i)$. In consequence, we can treat both $i$ and $j$ as receiving utility $U(x_i,y_j)$ from being matched. We refer to this situation as a market with \textbf{aligned preferences}---a setting in which the potential, a common utility function, describes the payoff to both sides of each matched pair.

\smallskip
The two matching mechanisms, DA and utilitarian welfare maximization, may appear fundamentally different. To analyze them within a unified framework, we embed both into a broader family of matchings. Specifically, we observe that the stable and welfare-maximizing matchings are two points on a spectrum reflecting planners' differing preferences over inequality. Formally,
let $\pi$ be a joint distribution on types $x$ and $y$ induced by a matching. Consider a planner who aims to minimize
a CARA version of Atkinson's index of inequality
\[
\text{Atkinson}_\alpha(\pi)=\int_{X \times Y} \frac{1-\exp\left(\alpha\cdot  U(x, y)\right)}{\alpha} \, \diff\pi(x, y)
\]
over all matchings. Minimizing the Atkinson index is equivalent (up to an additive constant and a sign) to maximizing the sum of CARA-transformed aligned utilities $U$. 
The parameter \( \alpha \) governs the planner’s preferences over inequality: when \( \alpha < 0 \), the objective is concave in utility, and the planner is inequality-averse; when \( \alpha > 0 \), the objective is convex, and the planner is inequality-loving. As \( \alpha \to 0 \), the objective converges to utilitarian welfare.

We show in Theorem \ref{th_equivalence} that the stable matching corresponds to the limit of the planner’s solution as \( \alpha \to +\infty \). In this limit, the planner places almost all weight on the highest-utility matches and little on the rest. This characterization explains the long left tail of low-utility matches in Figure \ref{fig:simulations_welfare}: it reflects the extreme preference for inequality of the limiting Atkinson index. At the same time, it adds to the puzzle that the overall welfare distributions under DA and utilitarian matching are so closely aligned. Theorem \ref{th_welfare} helps to explain this tension by showing that the limiting Atkinson objective nevertheless attains at least one-half of the utilitarian welfare. So the similarity is not coincidental, despite the sharp differences in the objectives.

\section{Stability and inequality}\label{sec_epsiloin_stable_general}

We study a general two-sided matching market.
Each side of the market is modeled as a population of \textbf{types}.

Let $X$ and~$Y$ be two sets of agent types. The type of an agent encodes all the relevant information about the agent. We assume that~$X$ and~$Y$ are complete separable metric spaces endowed with their Borel $\sigma$-algebras. The sets $X$ and $Y$ can be finite or infinite. The populations of agents' types over~$X$ and~$Y$ are represented by~$\mu\in \Delta(X)$ and~$\nu\in \Delta(Y)$, where $\Delta(Z)$ denotes the set of probability distributions over $Z$. Since  $\mu(X) = \nu(Y) = 1$, the total mass of agents on each side is equal; that is, the market is balanced. This simplifying assumption is not essential, and unbalanced markets can be accommodated by introducing ``dummy'' agents (see Example~\ref{ex_organ} in the Supplementary Appendix). 
Each $x\in X$ has a complete and transitive preference $\succeq_x$ over $Y$,
and each $y\in Y$ has a complete and transitive preference $\succeq_y$ over
$X$. The classical marriage market of \citet{gale1962college} arises as the
special case in which $X$ and $Y$ are finite and $\mu,\nu$ are counting measures normalized to have total mass $1$. More generally, any market with atomic $\mu$ and $\nu$ having finitely many atoms of equal mass reduces to the classical finite model.

We say that the preferences of agents in~$X$ and~$Y$ are \textbf{aligned} if there exists a function \(U\colon X\times Y\to \R\) such that  
\begin{equation}\label{eq_potential}
        \begin{split}
            U(&x,y)\geq U(x,y')\Longleftrightarrow y\succeq_x y' \\
            U(&x,y)\geq U(x',y)\Longleftrightarrow x\succeq_y x' 
        \end{split}\qquad \mbox{for all}\quad  x,x',y,y'.
    \end{equation} 
We call~$U$ the \textbf{potential} of the market. The potential is a single utility function serving as a shared cardinal representation of preferences for each side of the market. Its existence captures the idea that both sides evaluate matches according to a common ``match quality.''\footnote{For given ordinal preferences $(\succeq_x)_{x\in X}$ and $(\succeq_y)_{y\in Y}$, 
the existence of a potential is studied by 
\cite*{niederle2009decentralized} and \cite*{ferdowsian2020decentralized}, who establish its 
connection to potential games \citep*{monderer1996potential}. 
Specifically, $U$ is a potential representing both sides of the market if and only if it 
serves as a common utility function for a two-player game in which one player chooses from 
$X$, the other from $Y$, and the preference of the first player over their own actions 
given the second player's choice of $y$ is $\succeq_y$, and vice versa. 
\cite{voorneveld1997characterization} characterizes the preference families 
$(\succeq_x)_{x\in X}$ and $(\succeq_y)_{y\in Y}$ for which such a common utility exists, 
that is, for which this game is an ordinal potential game. Their condition translates directly 
to the market setting, taking the form of the absence of preference cycles.
}

For example, $X,Y\subset \R$ and $U(x,y)=-|x-y|$ represent a spatial matching market on 
the real line, where agents on both sides prefer to be matched with partners as close to themselves as 
possible. This market illustrates that aligned preferences do not imply identical preferences: two 
distinct types of agents $x,x'\in X$ will generally disagree on how they rank agents in~$Y$, even though 
their evaluations are based on the same underlying match quality~$U$. We further explore this spatial 
model and its multidimensional generalizations in Section~\ref{sec:applications}.

In what follows, we focus on aligned markets and treat the potential $U$ as a
primitive of the model, often referring to it simply as \textbf{utility}. We assume
that $U$ is a bounded measurable function and, when needed, impose additional regularity conditions.\footnote{Any strictly increasing reparameterization of the potential is again a potential. For instance, if $U$ is unbounded, then $\tilde{U}=\arctan(U)$ is a bounded potential representing the same market.} We return to the substance of the alignment assumption later in the
paper: Section~\ref{sec_approx_aligned} shows that even significantly
misaligned markets can often be approximated by aligned ones.

\paragraph{Stability and Fairness.}
Any pairing of agents in the market induces a joint distribution $\pi$ of
types. This distribution belongs to $\Delta(X\times Y)$ and has $\mu$ and $\nu$
as its marginals on $X$ and $Y$, respectively. We refer to any such distribution
as a \textbf{matching} and denote by $\Pi(\mu,\nu)$ the set of all matchings:
\[
\Pi(\mu,\nu)
  = \Big\{\,\pi \in \Delta(X\times Y)
      \ \text{with marginals}\ 
      \mu\ \text{and}\ \nu\,\Big\}.
\]
We allow for a given type $x$ to be associated with a distribution over partners
$y$, and similarly for each $y$ to be associated with a distribution over $x$.\footnote{Since we model matching at the level of types, it is natural to allow
for this randomization: even when individual agents are paired deterministically,
different agents of the same type may end up with partners of different types.}

Recall the classical notion of stability. For any two formed pairs
$(x_1,y_1)$ and $(x_2,y_2)$, stability requires that at least one of the
following holds:
\[
y_1 \succeq_{x_1} y_2
\qquad\text{or}\qquad
x_2 \succeq_{y_2} x_1.
\]
Otherwise, the pair $(x_1,y_2)$ would block the matching,
as both $x_1$ and $y_2$ would strictly prefer each other to their current
partners. Under aligned preferences, this condition is equivalent to requiring that
\begin{equation}\label{eq_stable_common}
    U(x_1,y_2)
    \leq
    \max\bigl\{ U(x_1,y_1),\, U(x_2,y_2) \bigr\}.
\end{equation}
This motivates the following definition. A matching $\pi$ is \textbf{stable} if, for $\pi\times\pi$-almost all pairs
$(x_1,y_1)$ and $(x_2,y_2)$, inequality~\eqref{eq_stable_common} holds.\footnote{For finite $X$ and $Y$ with counting measures, any stable matching according to this definition can be represented as a convex combination of deterministic stable matchings. Indeed, by the Birkhoff-von-Neumann theorem, any matching $\pi$ of such populations is a convex combination of deterministic ones. Every pair of couples $(x_1,y_1)$ and $(x_2,y_2)$ matched in one of these deterministic matchings has positive weight in $\pi$.
Hence, if $\pi$ is stable according to~\eqref{eq_stable_common}, so are all the deterministic matchings in the convex combination. 

Similar notions of stability  were introduced by \cite*{echenique2013revealed}, \cite*{kesten2015theory}, and (in full generality) \cite*{greinecker2021pairwise}.}
We also consider approximate stability. A matching $\pi$ is
$\varepsilon$-\textbf{stable} with $\varepsilon\ge0$ if for~$\pi\times \pi$-almost all pairs~$(x_1,y_1)$ and~$(x_2,y_2)$
\begin{equation}\label{eq_epsilon_stable_common}
    U(x_1,y_2)\leq \max\left\{U(x_1,y_1),\, U(x_2,y_2)\right\}+ \varepsilon.
\end{equation} 
 In other words, there is no potential blocking pair in which both agents can improve by more than $\varepsilon$.

If $U$ is continuous, our notion of stability admits a
more intuitive pointwise formulation. The support of~$\pi$, denoted
by~$\supp(\pi)$, is a minimal closed set of full measure. For continuous~$U$,
Lemma~\ref{lm_deterministic_equivalence} in Supplementary Appendix~\ref{app_lemma_pointwise_stability} shows that our definitions of stability are equivalent to requiring inequalities~\eqref{eq_stable_common} or \eqref{eq_epsilon_stable_common} to hold for all $(x_1,y_1),(x_2,y_2)\in\supp(\pi)$.

\medskip
We now turn to our notion of fairness. It reflects a concern for the agents who are worst off. To each matching~$\pi$ we assign
a number~$\mathcal{U}_{\min}(\pi)$ equal to the minimal utility of a couple matched under~$\pi$
and denote by $\mathcal{U}_{\min}^*(\mu,\nu)$ the best~$\mathcal{U}_{\min}(\pi)$ achievable over all matchings $\pi$.\footnote{Under the assumption that~$X$ and~$Y$ are compact and~$U$ is continuous,~$\mathcal{U}_{\min}$ and~$\mathcal{U}_{\min}^*$ are well-defined. Indeed, the support of~$\pi$ is a closed set, 
a closed subset of a compact space is compact, and hence the minimum is attained; the fact that the maximum is achieved is established below in Corollary~\ref{cor_existence} even for non-compact~$X$ and~$Y$. If~$U$ is discontinuous or~$X,Y$ are not compact, 
we replace the minimum with the essential infimum and the maximum with the supremum.}  
$$\mathcal{U}_{\min}(\pi)=\min_{(x,y)\in \supp(\pi)} U(x,y)\qquad\text{and}\qquad \mathcal{U}_{\min}^*(\mu,\nu)=\max_{\pi\in \Pi(\mu,\nu)} \mathcal{U}_{\min}(\pi).$$ This quantity $\mathcal{U}_{\min}^*(\mu,\nu)$ is the egalitarian lower bound in the spirit of~\cite*{rawls1971atheory}: the highest utility level feasible for every couple in the population simultaneously.

    A matching~$\pi\in \Pi(\mu,\nu)$ is~$\varepsilon$-\textbf{egalitarian} if there is a subset~$S\subset X\times Y$ with~$\pi(S)\geq 1-\varepsilon$ such that
    $$U(x,y)\geq \mathcal{U}_{\min}^*(\mu,\nu)-\varepsilon\qquad \mbox{for any }\quad (x,y)\in S.$$
In other words, for a large fraction of the population, the utility resulting from an~$\varepsilon$-egalitarian matching almost satisfies the egalitarian lower bound.  For~$\varepsilon=0$, we will refer to~$0$-egalitarian matchings simply as \textbf{egalitarian}.

\subsection{Stability, inequality, and optimal transport}\label{sec:stabilityandOT}

Our main results describe a link between stability, fairness, and pursuing a particular welfarist objective with different attitudes toward inequality.

Consider a matching market with aligned preferences represented by a potential~$U\colon X\times Y\to\R$. To each matching $\pi$ and each $\alpha\ne0$, we assign a CARA version of the index of inequality introduced by~\cite{atkinson1970measurement}:
\begin{equation}\label{eq_Atkinson}
    \text{Atkinson}_\alpha(\pi)=\int_{X \times Y} \frac{1-\exp\left(\alpha\cdot  U(x, y)\right)}{\alpha} \, \diff\pi(x, y).
\end{equation}
Up to an additive constant and a sign, this index equals the sum of CARA-transformed utilities. 
Consider a planner who wants to minimize Atkinson's inequality index:
\begin{equation}\label{eq_Atkinson_tranport}
    \min_{\pi\in \Pi(\mu,\nu)} \text{Atkinson}_\alpha(\pi).
\end{equation}
Because the integrand in~\eqref{eq_Atkinson} is strictly decreasing in
$U(x,y)$ for every $\alpha$, any minimizer is Pareto-efficient.
The distributional implications, however, differ with the sign of~$\alpha$.
When $\alpha>0$, the index is primarily driven by high-utility couples; minimizing
it therefore favors inequality by prioritizing the well-being of
high-utility couples over that of low-utility ones. By contrast, when
$\alpha<0$, the objective is driven by low-utility
couples, favoring more egalitarian outcomes. As $\alpha\to 0$, the objective
becomes neutral to the distribution of welfare: the integrand converges to
$-U(x,y)$, so we define
\begin{equation}\label{eq_welfare}
  W(\pi)=\int_{X\times Y}U(x,y)\diff\pi(x,y) \qquad\text{ and }\qquad 
  \text{Atkinson}_0(\pi)=-W(\pi).
\end{equation}
Hence, for $\alpha=0$, the planner's problem 
becomes
maximization of utilitarian social welfare.

Our first result shows that varying $\alpha$ from $-\infty$ to $+\infty$
interpolates between the \emph{fairness}, \emph{welfare}, and \emph{stability}
goals. 
\begin{theorem}\label{th_equivalence} Let~$\pi^*$ be a solution to~\eqref{eq_Atkinson_tranport}.
\begin{enumerate}
\item If~$\alpha>0$, then~$\pi^*$ is~$\varepsilon$-stable, with~$\varepsilon=\frac{\ln 2}{\alpha}$.
\item If~$\alpha=0$, then~$\pi^*$ maximizes utilitarian welfare.
\item If~$\alpha<0$, then~$\pi^*$ is~$\varepsilon$-egalitarian, with~$\varepsilon=\frac{\max\{1,\, \ln |\alpha|\}}{|\alpha|}$.
\end{enumerate}
\end{theorem}

The stability and fairness goals correspond to opposite sides of the
$\alpha$-spectrum. Stable matchings are selected by a planner with little concern for fairness, favoring high-utility matches at the expense of producing low-utility matches. 
In Section~\ref{sec:applications} we illustrate the conclusion of Theorem~\ref{th_equivalence} in a model of spatial matching, showing how the inequality obtained as a by-product of stability can be very significant and (arguably) unexpected.
While the conclusion of Theorem~\ref{th_equivalence} is specific to aligned markets, the underlying
tension is more general: as argued in Section~\ref{sec_approx_aligned},
aligned markets can approximate empirically relevant non-aligned settings.
\medskip

A technical contribution of the theorem is to establish a link between optimal transport theory and common market design goals.\footnote{See
\cite{galichon2018optimal} or  \cite{ekeland2010notes} for an introduction to
optimal transport for economists.} Minimizing Atkinson's index corresponds to a
particular optimal transport problem, and the tools we borrow from optimal
transport underlie Theorem~\ref{th_equivalence} and the subsequent analysis.

The canonical optimal transport problem assumes a measurable cost function
$c\colon X\times Y\to\R$ and marginal distributions $\mu\in\Delta(X)$ and
$\nu\in\Delta(Y)$. The problem is to find a matching $\pi$ that minimizes
total cost:
\begin{equation}\label{eq_transport}
  \min_{\pi\in\Pi(\mu,\nu)}\int_{X\times Y} c(x,y)\diff\pi(x,y).
\end{equation}
In particular, minimizing~\eqref{eq_Atkinson_tranport} is an optimal
transport problem with cost
\begin{equation}\label{eq_c_alpha}
  c_\alpha(x,y)=
  \begin{cases}
    \dfrac{1}{\alpha}\bigl(1-\exp(\alpha\cdot  U(x,y))\bigr), & \alpha\neq 0,\\[4pt]
    -\,U(x,y), & \alpha=0.
  \end{cases}
\end{equation}
The proof of Theorem~\ref{th_equivalence} appears in
Appendix~\ref{app_proof_equivalence}. Part~1 relies on the fundamental
``monotonicity'' (more precisely, $c$-\emph{cyclical monotonicity})
property of solutions to optimal transport problems. Given
$c\colon X\times Y\to\R$, a set $\Gamma\subset X\times Y$ is called
$c$-\emph{cyclically monotone} if, for all $n\ge 2$ and
$(x_i,y_i)\in\Gamma$,
\begin{equation}\label{eq_c_monotone}
  \sum_{i=1}^n c(x_i,y_i)\ \le\ \sum_{i=1}^n c(x_i,y_{i+1}),
  \qquad\text{with } y_{n+1}=y_1.
\end{equation}
Solutions to optimal transport are supported on $c$-cyclically monotone
sets.\footnote{See \cite*{kausamo202360} for a survey.} Using a result of
\cite*{beiglbock2009optimal}, we show that the stability condition
\eqref{eq_epsilon_stable_common} follows from $c$-cyclical monotonicity.
To prove part~3, we show that the contribution of low-utility couples to the
transport objective becomes dominant for negative $\alpha$, so $\pi^*$ cannot
place substantial weight on such couples.

\smallskip
The existence of the optimal matching $\pi^*$ in Theorem~\ref{th_equivalence} is guaranteed under continuity and boundedness assumptions on $U$. No compactness of $X$ or $Y$ is required: a
solution to~\eqref{eq_transport} exists whenever $c$ is lower semicontinuous
and bounded from below \citep*[Theorem~4.1]{villani2009optimal}. In particular,  it is enough that $U$ be bounded from above when~$\alpha>0$ and bounded from below when~$\alpha<0$.

Under these assumptions on $U$, Theorem~\ref{th_equivalence} implies the
existence of {exact} stable and egalitarian matchings, obtained as weak
limits as $\alpha\to\pm\infty$. Denote by $\Pi_{+\infty}^U(\mu,\nu)$ the set
of matchings $\pi$ that arise as weak limits
$\pi=\lim_{n\to+\infty}\pi_{\alpha_n}$, where each $\pi_{\alpha_n}$ solves
\eqref{eq_transport} with the cost $c_{\alpha_n}$ in~\eqref{eq_c_alpha} for
some sequence $\alpha_n\to+\infty$. Similarly, let $\Pi_{-\infty}^U(\mu,\nu)$
be the set of weak limits for some sequence $\alpha_n\to-\infty$.

\begin{corollary}\label{cor_existence}
  For continuous and bounded $U$, the sets
  $\Pi_{+\infty}^{U}(\mu,\nu)$ and $\Pi_{-\infty}^{U}(\mu,\nu)$ are
  nonempty and weakly closed. Every matching in
  $\Pi_{+\infty}^{U}(\mu,\nu)$ is stable, and every matching in
  $\Pi_{-\infty}^{U}(\mu,\nu)$ is egalitarian.
\end{corollary}

Corollary~\ref{cor_existence} provides the existence of stable and egalitarian matchings. For finite or discretized markets, Corollary~\ref{cor_existence} also suggests a computational method for finding these matchings: solve a sequence of finite linear programs for increasing values of \(|\alpha|\) and study their limiting solutions. 
The existence of stable matchings for markets with aligned preferences gives a special case of the general existence result proven in \cite*{greinecker2021pairwise}. In Supplementary Appendix~\ref{app_corr_existence}, 
we show that Corollary~\ref{cor_existence} follows from Theorem~\ref{th_equivalence} via a standard compactness argument.

For finite markets, the transportation problem in Theorem~\ref{th_equivalence} is solved by a deterministic~$\pi$. Indeed, for finite $X$ and~$Y$ with normalized counting measures, by the Birkhoff-von-Neumann theorem, the extreme points of the matching polytope are deterministic matchings. A linear objective is optimized at an extreme point, so there exist deterministic matchings that are  welfare-maximizing or egalitarian. To obtain a deterministic exact stable matching, we may use a large but finite value of~$\alpha$. Indeed, suppose that there are no payoff ties from changing partners, and denote by~$\delta>0$ the smallest utility gap arising from any such change:
$\delta=\min\left\{\min_{x,\ y\ne y'}|U(x,y)-U(x,y')|, \ \min_{y,\ x\ne x'}|U(x,y)-U(x',y)|\right\}.$
Then any~$\varepsilon$-stable matching with~$\varepsilon<\delta$ is automatically stable. Combining this observation with Theorem~\ref{th_equivalence}, we conclude that solutions of the optimal transportation problem with~$\alpha> \frac{\ln 2}{\delta}$ are stable.
\medskip

The conclusions of Theorem~\ref{th_equivalence} can be generalized in several ways. 
First, the result is not specific to the exponential objective~$c_\alpha$. As we demonstrate in Appendix~\ref{app_proof_equivalence}, it is enough to take
\begin{equation}\label{eq_general_cost}
c(x,y)=-h\left(U(x,y)\right),
\end{equation}    
where~$h$ is a monotone-increasing function with a large logarithmic derivative. 
Indeed, we show that if~$h>0$ is monotone increasing and
$\frac{h'(t)}{h(t)}\geq \alpha$
for some positive~$\alpha$ and all~$t$ in the range of~$U$, then any solution to the transportation problem with cost~\eqref{eq_general_cost} is~$\varepsilon$-stable with~$\varepsilon=\frac{\ln 2}{\alpha}$. Similarly, if~$h<0$ is monotone increasing and satisfies~$\frac{h'(t)}{h(t)}\leq\alpha$ with~$\alpha<0$, then the solution to the transportation problem is~$\varepsilon$-egalitarian with~$\varepsilon=\frac{\max\{1,\, \ln |\alpha|\}}{|\alpha|}$. 

For example, Atkinson's CARA index used in Theorem~\ref{th_equivalence} corresponds to $h(t)=\frac{1}{\alpha}e^{\alpha t}$ for $\alpha\neq 0$, in which case the above inequalities hold with equality. One can similarly consider a CRRA analogue, taking $h(t)=\frac{1}{\beta}t^{\beta}$ for $\beta\neq 0$ on the domain $t>0$. If $U$ takes values in an interval $[\underline{U},\overline{U}]$ with $\underline{U}>0$, then the above inequalities hold with $\alpha=\beta/\overline{U}$. Consequently, a matching minimizing the CRRA index is $\left(\frac{\overline{U}\,\ln 2}{\beta}\right)$-stable when $\beta>0$, and $\left(\frac{\overline{U}\,\max\{1,\,\ln|\beta/\overline{U}|\}}{|\beta|}\right)$-egalitarian when $\beta<0$. In contrast to the CARA case, the resulting approximation bound depends on the range of the potential $U$. This dependence is the primary reason we focus on the CARA-based choice of $h$ in Theorem~\ref{th_equivalence}.

Second, Theorem~\ref{th_equivalence} carries over (after only minor tweaks) to multi-sided matching markets, as we show in Supplementary Appendix~\ref{sec_multi_side}. This stands in stark contrast to the conventional wisdom that multi-sided markets are far less tractable than their two-sided counterparts.

\subsection{Welfare and fairness of stable matchings}\label{sec_welfare_and_fairness}

Theorem~\ref{th_equivalence} reveals that fairness and stability sit at opposite ends of the~$\alpha$-spectrum: exact fairness at~$-\infty$ and exact stability at~$+\infty$. At first glance, this dichotomy may suggest that stability offers no protection for individual fairness. Yet, as we show next, despite potential losses in both fairness and welfare, stability still comes with a nontrivial guarantee.

To understand why losses in welfare and fairness are bounded, we reinterpret the definition of stability. Consider a continuous utility~$U$ and a matching~$\pi\in \Delta(X\times Y)$ with marginals~$\mu$ and~$\nu$. By Lemma~\ref{lm_deterministic_equivalence}, matching~$\pi$ is~$\varepsilon$-stable with~$\varepsilon\geq 0$ if for all~$(x_1,y_1)$ and~$(x_2,y_2)$ in the support of~$\pi$,  inequality~\eqref{eq_epsilon_stable_common} holds, i.e., 
    $U(x_1,y_2)\leq \max\left\{U(x_1,y_1),\, U(x_2,y_2)\right\}+ \varepsilon.$
The couple~$(x_1,y_2)$ can serve as a generic element of the product space~$X\times Y$. This observation allows us to interpret formula~\eqref{eq_epsilon_stable_common} as follows: 
\begin{quote}
{\normalsize{$\pi$ is~$\varepsilon$-stable if, for a generic couple~$(x,y)$---not necessarily matched under $\pi$---the utility of at least one of the partners $x$ or $y$ in their respective match under $\pi$ is at least the utility of the hypothetical match $(x,y)$ minus~$\varepsilon$.}}
\end{quote}
Since the utility of a hypothetical couple~$(x,y)$ provides a lower bound to the utility of at least one of the partners in a stable match, the latter utilities cannot be too small simultaneously. This observation bounds the extent to which welfare and fairness must be sacrificed to obtain stability.

Recall that the utilitarian welfare~$W(\pi)$ of a matching~$\pi$ is given by~\eqref{eq_welfare} and 
let~$W^*(\mu,\nu)$ be the maximal welfare over all~$\pi\in\Pi(\mu,\nu)$. We prove the following result in Appendix~\ref{app_proof_equivalence}.
\begin{theorem}\label{th_welfare}
For bounded continuous utility~$U\geq 0$, any~$\varepsilon$-stable matching~$\pi$ satisfies\footnote{If a bounded continuous utility $U$ does not satisfy the non-negativity requirement, one can apply Theorem~\ref{th_welfare} to $\tilde{U}(x,y)=U(x,y)-\inf_{x',y'}U(x',y')$. For example, this argument applies to the distance-based utility of Section~\ref{sec:applications} if $\mu$ and $\nu$ have bounded support.}
$$W(\pi)\geq \frac{1}{2}\left(W^*(\mu,\nu)-\varepsilon\right).$$
Moreover,~$\pi$ is~$\varepsilon'$-egalitarian with
$\varepsilon'=\max\left\{\frac{1}{2},\, \varepsilon\right\}.$
\end{theorem}

In particular, we obtain that any stable matching guarantees~$1/2$ of the optimal welfare and is~$1/2$-egalitarian. For markets with a finite number of agents, we recover the welfare guarantee obtained by
\cite*{anshelevich2013anarchy}.
The bounds in Theorem~\ref{th_welfare} are conservative as they target~$\varepsilon$-stable matchings that have the lowest welfare or that are least egalitarian.

\subsection{Markets without aligned preferences}\label{sec_approx_aligned}

Thus far, our study of stability and fairness assumes aligned preferences. We now extend our results to certain non-aligned markets. 

First, it is worth remarking on the simple observation that our results carry over to approximately aligned markets with only minor modification. Specifically, if each $x$ and $y$'s utility from matching is within $\ep>0$ of an aligned utility $U(x,y)$, then any matching that is $\ep$-stable for the aligned market with utility $U$ is automatically $3\ep$-stable in the non-aligned market.\footnote{Let the utility of $x\in X$ be $v_x:Y\to\R$ and suppose $|v_x(y)-U(x,y)|\leq \ep$ for all $y$; assume analogously for the utilities of agents on the $Y$-side. If $x$ is matched with $y$ and considers a potential blocking pair with $y'$, then
$v_x(y')-v_x(y)\leq U(x,y')-U(x,y)+2\ep.$
Thus, whenever aligned $\ep$-stability implies \(U(x,y')\le U(x,y)+\ep\), the gain of $x$ is at most \(3\ep\); the analogous argument applies to the $Y$-side.}

Second, and more importantly, we show that in certain large finite markets, even dramatically misaligned preferences can yield a market that is well approximated by a related aligned market. 
We assume that utilities have an aligned component and an idiosyncratic component, but make no assumption on the relative importance of these two features in preferences. 

Let $X=Y\subset \R$ be intervals and consider two finite populations of $n$ agents $X_n=\{x_1,\ldots, x_n\}\subset X$ and $Y_n=\{y_1,\ldots, y_n\}\subset Y$. We assume that~$X_n$ and~$Y_n$ are i.i.d.\ samples from non-atomic distributions $\mu\in \Delta(X)$ and $\nu\in \Delta(Y)$.
If a pair $(x_i,y_j)\in X_n\times Y_n$ is formed, agents $i$ and $j$ enjoy utilities
\begin{align*}
    u_{i}(j)=q(x_i,y_j)+\xi_{i}(j)\qquad\text{and}\qquad
v_{j}(i)=q(x_i,y_j)+\eta_{j}(i).
\end{align*}
Here, $q\colon X\times Y\to \R$ is a continuous function capturing the aligned component of agents' preferences. The idiosyncratic components $\xi_{i}(j)$ and $\eta_{j}(i)$ are independent shocks with continuous distribution functions $F_i$ and $G_j$. Importantly, the idiosyncratic components may be very large relative to the aligned component~$q$. They could be the main driver of agents' preferences.
\begin{theorem}\label{th_idiosyncractic}
For $\pi \in \Pi(\mu,\nu)$, there is a sequence  $\delta_n\to 0$ such that,
with probability at least $1-\delta_n$, there exists a deterministic matching $\pi_n$ of $X_n$ and $Y_n$ with
$$\left|\frac{\left|\Big\{(x_i,y_j)\in [a,b]\times[c,d]\ \colon \text{$x_i$ and $y_j$ are matched in $\pi_n$} 
\Big\}\right|}{n}-\pi\big([a,b]\times [c,d]\big)\right|\leq \delta_n$$
for all intervals $[a,b] \subseteq X$, $[c,d] \subseteq Y$. Moreover, for all \(x_i\) and \(y_j\) matched under \(\pi_n\),
\[
F_i\big(\xi_i(j)\big)\ge 1-\delta_n
\quad\text{and}\quad
G_j\big(\eta_j(i)\big)\ge 1-\delta_n .
\]
\end{theorem}
Theorem~\ref{th_idiosyncractic} establishes that the finite non-aligned market can be approximated by a large aligned market. Specifically, for any matching $\pi$ in the large market, there exists a finite-market matching $\pi_n$ that is close to $\pi$ while ensuring that individual agents receive close to their ``maximum possible'' idiosyncratic utilities, i.e., they are matched with a partner representing a large quantile in their idiosyncratic distribution.

The theorem does not assume that $\pi$ is stable or egalitarian. It can be used to approximate a matching that optimizes any welfarist objective. When $\pi$ is stable, $\pi_n$ inherits approximate stability. Indeed, all agents in $\pi_n$ are matched in a way that guarantees them a very high idiosyncratic utility. Consequently, any blocking pair in $\pi_n$ would need to rely on exploiting the aligned component of preferences, implying the existence of a corresponding blocking pair in the original stable matching $\pi$.

In Appendix~\ref{app_proof_equivalence}, we present a stronger result, Theorem~\ref{th_idiosyncratic_appendix}, which provides explicit formulas for the approximation error and does not rely on i.i.d.\ sampling of the population. So it applies to any---possibly deterministic---set of points whose empirical distributions are close to $\mu$ and $\nu$.

The idea that idiosyncratic components can effectively be ignored in large markets originates in \cite{lee2016incentive}. He studies markets in which all agents share the same common preference (corresponding to $q(x,y)=q_1(x)+q_2(y)$) perturbed by idiosyncratic shocks, and shows that stable matchings are approximately assortative and that most agents attain nearly maximal idiosyncratic payoffs. While both results capture aspects of the same phenomenon, Theorem~\ref{th_idiosyncratic_appendix} differs in several respects.  It applies to arbitrary aligned components $q(x,y)$ and constructs a nearby deterministic finite matching in which \emph{every} rather than \emph{most} matched agent attains a near-top idiosyncratic payoff, so the aligned benchmark is relevant not only to stability but also to welfare and fairness comparisons. In addition, the explicit bounds on $\delta_n$ make the result applicable to finite markets. These differences necessitate new proof techniques, which combine the probabilistic method, Dvoretzky-Kiefer-Wolfowitz concentration, and an adaptation of arguments from the classical Erd\H{o}s-R\'{e}nyi-type perfect-matching construction.

\section{Spatial matching and other applications}\label{sec:applications}

We consider a benchmark spatial matching model with aligned preferences, where the two sides of the market $X$ and $Y$ are subsets of $\mathbb{R}^d$ with $d\geq 1$, and each agent prefers to be matched to a partner as close as possible. This is captured by the distance-based potential
\begin{equation}\label{eq_distance_based}
    U(x,y)=-\|x-y\|=-\sqrt{\sum_{i=1}^d(x_i-y_i)^2}.
\end{equation}
This setting can capture more general preferences for homophily in an abstract space of ``types,'' for example, preferences for partners with similar political views or educational backgrounds. 

The link between stability, fairness, welfare, and optimal transport in Theorem~\ref{th_equivalence} is completely dimension-free: it holds for any~$d$. But to dig deeper into the structure of spatial matchings, we must pin down a specific dimension. We begin with the one-dimensional case ($d=1$), where Theorem~\ref{th_equivalence}, combined with known results in optimal transport, enables us to explicitly construct optimal matchings. This simple setting sharpens and illustrates the core tension between stability, welfare, and fairness uncovered by the theorem.

For $d=1$, the cost function \eqref{eq_c_alpha} from the theorem specializes to
\begin{equation}\label{eq_c_alpha_R}
    c_\alpha(x,y)=\frac{1-\exp(-\alpha\cdot \vert x-y\vert)}{\alpha},
\end{equation}
which is a strictly convex function of~$\vert x-y\vert$ for $\alpha<0$ and strictly concave for $\alpha>0$. Transportation problems on the line with costs given by concave and convex functions of distance are well understood.

For a cost function~$c(x,y)=h(| x-y|)$ with strictly convex~$h$, the solution to~\eqref{eq_transport} is the assortative matching, which is uniquely optimal and does not depend on the particular form of~$h$. The assortative matching is supported on the curve $F_\mu(x)=F_\nu(y)$, where~$F_\mu$ and~$F_\nu$ are the cumulative distribution functions of type distributions~$\mu$ on $X$ and~$\nu$ on $Y$, respectively.

Thus, for any $\alpha<0$, the assortative matching is the unique solution to the optimal transportation problem. Taking weak limits and noting that as $\alpha\to 0$ the objective approaches the utilitarian objective, Theorem~\ref{th_equivalence} implies that the assortative matching is fair and welfare-maximizing.

\begin{corollary}\label{cor_no_tension}
For $X=Y=\mathbb{R}$ and non-atomic $\mu,\nu$ with bounded support, the assortative matching is egalitarian and welfare-maximizing.
\end{corollary}

For example, suppose that $X$ is the population of agents uniformly distributed on the interval $[-1,0]$, while $Y$ is the population uniformly distributed on $[0,1]$. We can visualize this market by plotting the distribution of $X$ above the horizontal axis, and the population of $Y$ below the horizontal axis.
\begin{figure}[ht]
    \centering
    \begin{subfigure}[t]{0.35\textwidth}
        \centering
        \begin{tikzpicture}[scale=0.5]
            \draw[black,very thick] (-5,0) -- (5,0);
            \filldraw[fill=lightgray, draw=black, very thick] (0,0) rectangle (-4.5,3);
            \filldraw[fill=lightgray, draw=black, very thick] (0,0) rectangle (4.5,-3);
            \draw[very thick, dashed, ->] (-0.1,0) arc (180:0:2.3);    
            \draw[very thick, dashed, ->] (-1.2,0) arc (180:0:2.3);
            \draw[very thick, dashed, ->] (-2.3,0) arc (180:0:2.3);
            \draw[very thick, dashed, ->] (-3.4,0) arc (180:0:2.3);
            \draw[very thick, dashed, ->] (-4.5,0) arc (180:0:2.3);
        \end{tikzpicture}
        \caption{Egalitarian}
        \label{fig:stablevfair2}
    \end{subfigure}
        \hfill
    \begin{subfigure}[t]{0.35\textwidth}
        \centering
        \begin{tikzpicture}[scale=0.5]
            \draw[black,very thick] (-5,0) -- (5,0);
            \filldraw[fill=lightgray, draw=black, very thick] (0,0) rectangle (-4.5,3);
            \filldraw[fill=lightgray, draw=black, very thick] (0,0) rectangle (4.5,-3);
            \draw[very thick, dashed, ->] (-0.5,0) arc (180:0:0.5);    
            \draw[very thick, dashed, ->] (-1.5,0) arc (180:0:1.5);
            \draw[very thick, dashed, ->] (-2.5,0) arc (180:0:2.5);
            \draw[very thick, dashed, ->] (-3.5,0) arc (180:0:3.5);
            \draw[very thick, dashed, ->] (-4.5,0) arc (180:0:4.5);
        \end{tikzpicture} 
        \caption{Stable}
        \label{fig:stablevfair1}
    \end{subfigure}
    \begin{subfigure}[t]{0.35\textwidth}
    \end{subfigure}
    \caption{Two solutions to the matching problem. }
    \label{fig:stablevfair}
\end{figure}

Figure~\ref{fig:stablevfair2} visualizes the assortative matching by connecting matched agents using half-circles. In this matching, all agents are matched with a partner who is one unit away. By Corollary~\ref{cor_no_tension}, this matching is both egalitarian and welfare-maximizing. Notice, however, that the assortative matching is not stable. Indeed, every pair $(-z,z)$ with $z\in(0,1/2)$ blocks this matching, as each partner prefers this match to their assortative partners $-z+1$ and $z-1$. 

Matching agents of type~$-z\in [-1,0]$ with those of type~$z\in [0,1]$ turns out to be the unique stable matching. It is illustrated in Figure~\ref{fig:stablevfair1}. 
Notice that the stable matching induces substantial inequality in the agents' utilities. Agents close to zero are paired with partners who have similar types to them, while agents far from zero are matched with partners of very different types. This stable matching leads to the same utilitarian welfare of $-1$ as the egalitarian one and, hence, in this example there is no tension between the welfare and stability goals. While the total welfare remains intact, the stability goal leads to dramatic inequality, increasing the utility of agents in $(-1/2, 1/2)$ at the expense of those in $(-1,-1/2)$ and $(1/2,1)$, thus creating substantial inequality with some couples enjoying utility of $0$ and others having utility of $-2$. As we will see, for more general markets, stability not only leads to dramatic inequality but also to overall welfare loss.

To analyze the structure of stable matchings in the general case, we note that the stable matching in Figure~\ref{fig:stablevfair1} features a key structural property: the half-circles indicating matched pairs do not cross. This no-crossing property plays a central role in the theory of one-dimensional optimal transport with concave costs \citep{mccann1999exact}.\footnote{According to \cite*{villani2009optimal}, the no-crossing property dates back to papers by Monge, who initiated the study of optimal transport; see also the discussion by~\cite*{boerma2023composite}.} In particular, the set of matchings satisfying the no-crossing property is well understood.

Formally, for any two real numbers~$z_1$ and~$z_2$, 
let~$O(z_1,z_2)$ denote the circle in~$\mathbb{R}^2$ whose diameter has endpoints~$(z_1,0)$ and~$(z_2,0)$. A matching~$\pi$ on~$\mathbb{R}$ satisfies the \textbf{no-crossing property} if, for any two pairs~$(x,y)$ and~$(x',y')$ in the support of~$\pi$, the circles~$O(x,y)$ and~$O(x',y')$ do not intersect unless~$x=x'$ or~$y=y'$.

\cite{mccann1999exact} (Theorem 3.11) showed that for any cost $c(x,y)=h(| x-y|)$ with strictly concave~$h$ and non-atomic $\mu$ and~$\nu$ on~$\mathbb{R}$, optimal matchings satisfy the no-crossing property. Specifically, the mass common to $\mu$ and $\nu$ is first eliminated (these agents are matched with their ideal partners $x=y$), and then the remaining disjoint populations are matched in a no-crossing fashion. Moreover, the set of all no-crossing matchings consists of a finite number of parametric families. 
\begin{corollary}\label{cor_no_crossing}
The optimal matchings for cost~\eqref{eq_c_alpha_R} with \(\alpha>0\) satisfy no-crossing. By taking weak limits as \(\alpha\to+\infty\), stable matchings obtained as such limits also satisfy no-crossing.\footnote{The set of no-crossing matchings with given marginals is closed under weak convergence since these matchings form a finite number of parametric families, continuous and compact.}
\end{corollary}
 In Supplementary Appendix~\ref{sec_algorithm}, Lemma~\ref{lm_stable_is_no_crossing} shows directly that every stable matching on the line satisfies no-crossing. To illustrate the power of this observation, consider the following example. Suppose~$X$ has density~$1/3$ on the interval~$[-2,-1]$ and density~$2/3$ on the interval~$[0,1]$. Assume that~$Y$ has density~$1/3$ on the intervals~$[-1,0]$ and~$[1,3]$. Figure~\ref{fig:examplemeasure} depicts the densities.
\begin{figure}[ht]
    \centering
    \begin{tikzpicture}[scale=0.55]
        \draw[black,very thick] (-6,0) -- (5,0);
        \filldraw[fill=lightgray, draw=black, very thick] (-5.5,0) rectangle (-3.5,2);
        \filldraw[fill=lightgray, draw=black, very thick] (-3.5,0) rectangle (-1.5,-2);
        \filldraw[fill=lightgray, draw=black, very thick] (-1.5,0) rectangle (0.5,4);
        \filldraw[fill=lightgray, draw=black, very thick] (0.5,0) rectangle (4.5,-2);
        \node at (-5.6,-0.5) {\footnotesize~$-2$};
        \node at (-1.2,-0.5) {\footnotesize~$0$};
        \node at (4.8,-0.5) {\footnotesize~$3$};
    \end{tikzpicture}
    \caption{The density is above the axis for $X$, and below for $Y$.}
    \label{fig:examplemeasure}
\end{figure}
No-crossing matchings in this example form a one-parameter family that can be described as follows. The parameter~$\theta\in [0,1]$ controls the proportion of agents in the region~$[-2,-1]$ who are matched with agents in the region $[1,3]$. Any value of $\theta$ between $0$ and $1$ can be chosen. Three specific values of~$\theta$ are shown in Figure~\ref{fig:runningexample}.
\begin{figure}[ht]
    \centering
    \begin{subfigure}{0.48\textwidth}
    \centering
    \begin{tikzpicture}[scale=0.55]
        \draw[black,very thick] (-6,0) -- (5,0);
        \filldraw[fill=lightgray, draw=black, very thick] (-5.5,0) rectangle (-3.5,2);
        \filldraw[fill=lightgray, draw=black, very thick] (-3.5,0) rectangle (-1.5,-2);
        \filldraw[fill=lightgray, draw=black, very thick] (-1.5,0) rectangle (0.5,4);
        \filldraw[fill=lightgray, draw=black, very thick] (0.5,0) rectangle (4.5,-2);

        \draw[very thick, dashed, ->] (-5.1,0) arc (180:0:1.6);
        \draw[very thick, dashed, ->] (-4.7,0) arc (180:0:1.2);
        \draw[very thick, dashed, ->] (-4.3,0) arc (180:0:0.8);
        \draw[very thick, dashed, ->] (-3.9,0) arc (180:0:0.4);

        \draw[very thick, dashed, ->] (-1.1,0) arc (180:0:2.7);
        \draw[very thick, dashed, ->] (-0.7,0) arc (180:0:1.9);
        \draw[very thick, dashed, ->] (-0.3,0) arc (180:0:1.1);
        \draw[very thick, dashed, ->] (0.1,0) arc (180:0:0.4);

        \node at (-5.6,-0.5) {\footnotesize~$-2$};
        \node at (-1.2,-0.5) {\footnotesize~$0$};
        \node at (4.8,-0.5) {\footnotesize~$3$};
    \end{tikzpicture}
    \caption{$\theta=0$}
    \end{subfigure}
       \begin{subfigure}{0.48\textwidth}
       \centering
     \begin{tikzpicture}[scale=0.55]
        \draw[black,very thick] (-6,0) -- (5,0);
        \filldraw[fill=lightgray, draw=black, very thick] (-5.5,0) rectangle (-3.5,2);
        \filldraw[fill=lightgray, draw=black, very thick] (-3.5,0) rectangle (-1.5,-2);
        \filldraw[fill=lightgray, draw=black, very thick] (-1.5,0) rectangle (0.5,4);
        \filldraw[fill=lightgray, draw=black, very thick] (0.5,0) rectangle (4.5,-2);
\draw[very thick, dashed, ->] (-1.215,0) arc (0:180:0.4);
\draw[very thick, dashed, ->] (-0.93,0) arc (0:180:0.85);
\draw[very thick, dashed, ->] (-0.645,0) arc (0:180:1.3);
\draw[very thick, dashed, ->] (-0.36,0) arc (180:0:1.3);
\draw[very thick, dashed, ->] (-0.07,0) arc (180:0:0.85);
\draw[very thick, dashed, ->] (0.21,0) arc (180:0:0.4);

\draw[very thick, dashed, ->] (-5.4,0) arc (180:0:4.9);
\draw[very thick, dashed, ->] (-5.05,0) arc (180:0:4.55);
\draw[very thick, dashed, ->] (-4.7,0) arc (180:0:4.2);
\draw[very thick, dashed, ->] (-4.3,0) arc (180:0:3.8);
\draw[very thick, dashed, ->] (-3.9,0) arc (180:0:3.4);
\draw[very thick, dashed, ->] (-3.55,0) arc (180:0:3.05);

        \node at (-5.6,-0.5) {\footnotesize~$-2$};
        \node at (-1.2,-0.5) {\footnotesize~$0$};
        \node at (4.8,-0.5) {\footnotesize~$3$};
    \end{tikzpicture}
        \caption{$\theta=1$}
    \end{subfigure}
       \begin{subfigure}{0.48\textwidth}
       \centering
    \begin{tikzpicture}[scale=0.55]
        \draw[black,very thick] (-6,0) -- (5,0);
        \filldraw[fill=lightgray, draw=black, very thick] (-5.5,0) rectangle (-3.5,2);
        \filldraw[fill=lightgray, draw=black, very thick] (-3.5,0) rectangle (-1.5,-2);
        \filldraw[fill=lightgray, draw=black, very thick] (-1.5,0) rectangle (0.5,4);
        \filldraw[fill=lightgray, draw=black, very thick] (0.5,0) rectangle (4.5,-2);
        
        \draw[very thick, dashed, ->] (-3.8,0) arc (180:0:0.3);
        \draw[very thick, dashed, ->] (-4.1,0) arc (180:0:0.6);
        
        \draw[very thick, dashed, ->] (-1.25,0) arc (0:180:0.35);
        \draw[very thick, dashed, ->] (-.95,0) arc (0:180:0.7);
        
        \draw[very thick, dashed, ->] (-0.4,0) arc (180:0:1.6);
        \draw[very thick, dashed, ->] (-0.1,0) arc (180:0:1);
        \draw[very thick, dashed, ->] (0.2,0) arc (180:0:0.5);
\draw[very thick, dashed, ->] (-5.4,0) arc (180:0:4.9);
\draw[very thick, dashed, ->] (-5.05,0) arc (180:0:4.55);
\draw[very thick, dashed, ->] (-4.7,0) arc (180:0:4.2);
\draw[very thick, dashed, ->] (-4.3,0) arc (180:0:3.8);

        \node at (-5.6,-0.5) {\footnotesize~$-2$};
        \node at (-1.2,-0.5) {\footnotesize~$0$};
        \node at (4.8,-0.5) {\footnotesize~$3$};
    \end{tikzpicture}
    \caption{$\theta=\frac{4}{7}$}
    \label{fig:runningexamplec}
    \end{subfigure}
    \caption{The matchings which satisfy no-crossing for three values of~$\theta$.}
    \label{fig:runningexample}
\end{figure}

The parametric structure of no-crossing matchings reduces the infinite-dimensional transportation problem with cost $c_\alpha(x,y)$ for $\alpha>0$ to determining a finite number of parameters. Unlike the case of $\alpha<0$, where the optimum is independent of $\alpha$, for $\alpha>0$ the optimum generally depends on $\alpha$ since the set of no-crossing matchings is not a singleton.
For example, Figure~\ref{fig:theta_alpha} illustrates the dependence of the optimal $\theta$ on $\alpha$ for the market from Figure~\ref{fig:runningexample}.
\begin{figure}[ht]
    \centering
   \includegraphics[width=8cm]{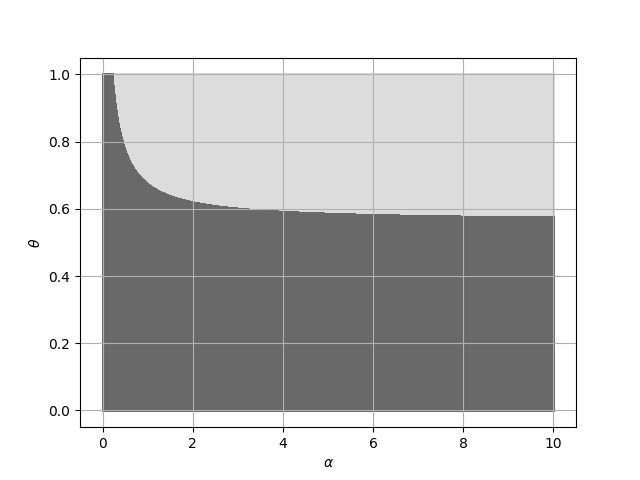}
    \caption{Dependence of $\theta$ on $\alpha$ in the optimal transportation problem with cost $c_\alpha$ and distribution from Figure~\ref{fig:runningexample}. The optimal choice of $\theta$ for each $\alpha$ is given by the curve where the two grey regions meet.}
    \label{fig:theta_alpha}
\end{figure}
By taking the limit as $\alpha\to +\infty$, we obtain that the stable matching corresponds to~$\theta={4}/{7}$.
In this stable matching, each type~$x$ is matched with~$y=h(x)$ where 
\begin{center}
   $ h(x) = \left\{\begin{array}{cc}
                            1-x, & x\in [-2,\ -1-\frac{3}{7}] \\
                              -2-x, & x\in(-1-\frac{3}{7},\ -1]\\
                            -2x, & x \in (0, \ \frac{2}{7}] \\
                            3-2x, & x \in (\frac{2}{7}, \ 1] \\
                \end{array}\right.$. 
\end{center}
Relative to the assortative matching, the stable matching again induces inequality of outcomes. In the assortative matching, the worst-off couple is matched with a partner at a distance of $2$ away. In the stable matching, any agent at $-2$ is matched with a partner at $3$, giving utility $-5$. 

This example reveals a stark fact: \emph{stable matchings can produce dramatically unfair outcomes}. Agents who are nearly indistinguishable can end up with radically different partners and utilities.

Take two agents, $x$ and $x'$, with types $-\frac{10}{7}-\varepsilon$ and $-\frac{10}{7}+\varepsilon$ for some $\varepsilon>0$. Agent $x$ is matched with a partner more than $3$ units away, while $x'$ is matched within $1$ unit. The result is that $x$ envies $x'$. In the language of stable matching theory, this is ``unjustified'' envy; but the label does not erase the underlying unfairness. Here, the inequity is driven by supply–demand imbalances rather than the usual considerations in matching theory. Agent $x$ is forced into a non-local match because of a scarcity of local partners; agents \(y\in[-1,0]\) have already formed matches with more preferred agents such as \(x'\). In this way, high-utility couples generate negative externalities for low-utility ones, amplifying inequality within an otherwise stable matching.

By taking the limit $\alpha\to 0$, we obtain that the no-crossing matching with $\theta=1$ maximizes welfare. The change in the optimal no-crossing matching as $\alpha$ ranges between $0$ and $+\infty$ reflects the tension between welfare and stability. 
Indeed, for the market from Figure~\ref{fig:runningexample}, the stable matching has a welfare of $-220/147\approx -1.5$, while the optimal welfare is $-4/3\approx -1.33$.

A simple sufficient condition for the absence of a stability-welfare tension is that \(\mu-\nu\) changes sign at most twice. Here, sign changes are counted after removing common mass and partitioning the real line into maximal intervals on which the signed measure \(\mu-\nu\) is nonnegative or nonpositive. For example, in the market from Figure~\ref{fig:stablevfair}, where the stable matching has optimal welfare, $\mu-\nu$ changes sign once at $0$, while in the market from Figure~\ref{fig:runningexample}, where there is a tension, there are three sign changes.

More generally, the stability-welfare tension is absent if $\mu-\nu$ changes sign at most twice. Indeed, \cite*{mccann1999exact} showed that for such $\mu$ and $\nu$, a no-crossing matching is unique (as in Figure~\ref{fig_two_times}).
Thus, the same no-crossing matching solves the transportation problem with the cost $c_\alpha$ from~\eqref{eq_c_alpha_R} for any $\alpha\in (0,+\infty)$. Letting $\alpha \to +\infty$, we conclude that this matching is stable. Letting $\alpha\to 0$, we obtain that it is also welfare-maximizing.
\begin{figure}[ht]
    \centering
		\begin{tikzpicture}[scale=0.45]
			\draw[black,very thick] (-10,0) -- (6,0);
			\filldraw[fill=lightgray, draw=black, very thick] (-4.5,0) rectangle (-9.0,3);
			\filldraw[fill=lightgray, draw=black, very thick] (2.25,0) rectangle (5.1,3);
			\filldraw[fill=lightgray, draw=black, very thick] (-4.5,0) rectangle (2.25,-3);
			
			\draw[very thick, dashed, ->] (-5,0) arc (180:0:0.5);    
			\draw[very thick, dashed, ->] (-6,0) arc (180:0:1.5);
			\draw[very thick, dashed, ->] (-7,0) arc (180:0:2.5);
			\draw[very thick, dashed, ->] (-8,0) arc (180:0:3.5);
			\draw[very thick, dashed, ->] (-9,0) arc (180:0:4.5);
			
			\draw[very thick, dashed, ->] (3.1,0) arc (0:180:0.5);    
			\draw[very thick, dashed, ->] (4.1,0) arc (0:180:1.5);
			\draw[very thick, dashed, ->] (5.1,0) arc (0:180:2.5);
			
		\end{tikzpicture}
    \caption{For $\mu-\nu$ changing sign two times, no-crossing matching is unique and thus is simultaneously stable and welfare-maximizing.}
    \label{fig_two_times}
\end{figure}

Problem complexity increases with the number of times $\mu-\nu$ changes sign. For at most two changes, a no-crossing matching is unique. For three changes, as in Figure~\ref{fig:theta_alpha}, we obtain a single one-parameter family. More generally, we obtain multiple parametric families whose number grows exponentially with the number of sign changes.\footnote{No-crossing matchings are parameterized by partitioning the real line into intervals where $\mu-\nu\geq 0$ or $\mu-\nu\leq 0$. For each interval of positivity, we specify the fraction of agents matched non-locally to the right ($\theta$ in Figure~\ref{fig:runningexample}) and the negativity intervals these matches come from. Thus, for each interval of positivity, we obtain a parameter and combinatorial data. The combinatorial data grows exponentially with the number of intervals.}

A designer pursuing stability can avoid the exploding search space by combining insights about no-crossing matchings with properties specific to stability. In Supplementary Appendix~\ref{sec_algorithm}, we describe an efficient algorithm that constructs a stable matching by sequentially identifying sub-markets similar to Figure~\ref{fig_two_times}, which are to be matched internally regardless of the rest of the population. If $\mu$ and $\nu$ have piecewise-constant density with at most $m$ intervals of constancy, the algorithm runs in time $O(m^2)$.

\begin{corollary}\label{cor_stability}
For \(X=Y=\mathbb{R}\) and non-atomic \(\mu,\nu\) such that \(\mu-\nu\) changes sign only finitely many times, there is a unique stable matching. If, in addition, \(\mu\) and \(\nu\) have piecewise-constant densities with at most \(m\) intervals of constancy, the stable matching can be computed in time \(O(m^2)\). For each \(x\in X\), there are at most two distinct types \(y,y'\in Y\) such that \((x,y)\) and \((x,y')\) are in the support. The welfare-stability tension---while present in general---disappears if \(\mu-\nu\) has at most two sign changes.
\end{corollary}

Corollaries~\ref{cor_no_tension} and~\ref{cor_stability} imply that for $\mu-\nu$ with at most two sign changes, there are egalitarian and stable matchings with optimal welfare. However, these two matchings are never the same unless $\mu$ coincides with $\nu$. More generally, for distinct non-atomic $\mu,\nu$, egalitarian and stable matchings cannot coincide since the former is increasing, while the latter contains decreasing regions (as any no-crossing matching).
\begin{corollary}\label{cor_stability_fairness_tension}
    For $X=Y=\mathbb{R}$ and distinct non-atomic $\mu,\nu$, there is always a tension between stability and fairness. That is, no matching is both stable and egalitarian.
\end{corollary}
\medskip

We now consider the multidimensional setting. For $\R^d$ with~$d>1$, there may be a tension between individual and collective welfare, absent in the one-dimensional case. 
\begin{example}[Fairness-welfare tension]
Let~$X=Y=\mathbb{R}^2$ and $U(x,y)=-\|x-y\|$. Denote by~$B_{\delta}(z)$ the ball of radius~$\delta$ centered at a point~$z\in \R^2$. Fix a small $\delta>0$. Let~\(\mu\) be uniform on \(B_{\delta}(0,0)\cup B_{\delta}(0,1)\), and~\(\nu\) be uniform on \(B_{\delta}(0,0)\cup B_{\delta}(1,0)\). In the welfare maximizing matching, the agents~$X$ and~$Y$ in~$B_{\delta}(0,0)$ are matched together and the agents in~$X$ from~$B_{\delta}(0,1)$ are matched with those from~$B_{\delta}(1,0)$ in~$Y$. As \(\delta\to0\), the welfare-maximizing matching has average distance approaching \(\frac{1}{2}\sqrt{2}\), and its longest matched distance approaches \(\sqrt{2}\). By contrast, the matching that pairs the \(X\)-agents in \(B_\delta(0,1)\) with the \(Y\)-agents in \(B_\delta(0,0)\) and the \(X\)-agents in \(B_\delta(0,0)\) with the \(Y\)-agents in \(B_\delta(1,0)\) has average distance approaching \(1\) and longest matched distance at most \(1+2\delta\).
\end{example}

Concave and convex optimal transport do not generally admit explicit solutions beyond~$d=1$, but some structural results can be obtained for~$d>1$. Now the cost function is
\begin{equation}\label{eq_c_alpha_R_d}
    c_\alpha(x,y)=\frac{1-\exp(-\alpha\cdot \|x-y\|)}{\alpha} \quad \text{for $\alpha\ne 0$}\qquad \text{and}\qquad  c_0(x,y)= \|x-y\|.
\end{equation} 
Say that a matching~$\pi$ is deterministic if it is supported on the graph of an invertible map~$s\colon X\to Y$. A matching~$\pi$ is diagonal if~$x=y$ for~$\pi$-almost all couples~$(x,y)$.
Off-the-shelf results from optimal transport (Theorems~1.2 and~1.4 of \cite*{gangbo1996geometry}) address multidimensional optimal transport with concave or convex costs. Their results imply the following corollary.
\begin{corollary}\label{cor:Rdstructure}
Consider a market with~$X=Y=\R^d$, compactly supported~$\mu$ and~$\nu$ absolutely continuous with respect to the Lebesgue measure, and potential~$U(x,y)=-\|x-y\|$. 
The optimal transportation problem with cost~\eqref{eq_c_alpha_R_d} admits an optimal matching~$\pi^*$ and the following assertions hold:
\begin{itemize}
    \item if~$\alpha<0$, then $\pi^*$ is unique and deterministic;
    \item if $\alpha>0$, then $\pi^*$ is unique and can be written as a
    convex combination of a deterministic matching and a diagonal one;
    \item if $\alpha=0$, then $\pi^*$ maximizes welfare, and there exists an optimal deterministic matching.
\end{itemize}
\end{corollary}
We note that the condition of absolute continuity can be weakened to requiring that~$\mu$ and~$\nu$ place zero mass on~$(d-1)$-dimensional surfaces. By imposing this on~$\mu$ only, we obtain a version of Corollary~\ref{cor:Rdstructure} where the conclusion that~$\pi^*$ is deterministic is replaced with the Monge property ($\pi^*$ is supported on a graph of~$s\colon X\to Y$, but~$s$ may not be invertible).

The egalitarian matching obtained as
the limit~$\alpha\to -\infty$ corresponds to the $L^\infty$-transportation problem \cite*{champion2008wasserstein,brizzi2023optimal}. The transportation literature has mostly focused on questions of existence, uniqueness, and cyclic monotonicity.
 For a distance-based cost, the~$L^\infty$-problem is interesting only for $d>1$ since, on the real line, it corresponds to the assortative matching by Corollary~\ref{cor_no_tension}. To the best of our knowledge, the limit problem for~$\alpha\to +\infty$ has not been studied prior to our paper.

\subsection{Discussion}

As we have argued, the spatial structure of a market is important in school choice. It has also been emphasized by a growing literature in international trade on buyer-seller networks; see, for instance, \cite*{chaney2014network,antras2017margins,panigrahi2021endogenous}. Here we discuss two other applications.

\medskip

\emph{Ride Sharing:} Ride-sharing platforms like Uber and Lyft operate large two-sided matching markets where thousands of riders are matched to drivers every day. While both riders and drivers, in principle, care about specific details of their match partner, the most important and salient feature is their physical distance. Drivers and riders both prefer to be matched with partners close to them to avoid wasteful travel time. Initially, Uber's matching algorithm was greedy: they simply matched each rider with the closest available driver. The resulting matching is stable in an aligned market with
$U(x,y)=-\text{distance}(x,y)$ as in~\eqref{eq_distance_based}. Quickly, however, Uber noticed that this greedy algorithm had some unappealing features.  The following is a quote from Uber's website describing the evolution of their matching algorithm:
\begin{quote}
{\normalsize{
    In the early days, a rider was immediately matched with the closest available driver. It worked well for most riders but sometimes led to long wait times for others. Across a whole city, those longer wait times really added up [...] In the seconds after a rider requests a ride, we evaluate nearby drivers and riders in one batch. We then pair riders and drivers in the distribution, aiming to reduce the average wait time for everyone, not just the closest pair. This helps keep things moving and rides reliable across the network \citep{Uber}.
    }}
\end{quote}
There are a few points worth emphasizing. First, Uber claims that the greedy match led to long wait times for some riders, while it worked well for most. This is consistent with the message that stability and fairness are in conflict (see Corollary~\ref{cor_stability_fairness_tension} and the long tail of low-utility matches in simulations reported in Figure~\ref{fig:simulations_welfare}). 
Second, Uber noticed that, not only did the greedy matching result in inequality, but the longer wait times ``added up.'' This suggests that the greedy algorithm produced overall longer total waiting times, which is consistent with the message in Theorem~\ref{th_equivalence}. The objective implicitly being optimized in a stable matching does not minimize total wait times. Theorem \ref{th_welfare} implies a bound on how poorly the stable matching can perform in terms of total wait time.

Uber now matches riders in batches. 
Within batches, the algorithm aims to reduce average wait times. That is, Uber does not use a stable matching, even within batches. 
\medskip 

\emph{Markets with transfers and limited commitment:}
Aligned preferences are a natural assumption in models of matching with transfers and limited commitment. We argue that alignment arises when agents cannot commit to utility transfers at the stage where they bargain over matches, and the post-match surplus is shared according to Nash bargaining. 

Consider the matching market with transferable utility introduced by \cite*{shapley1971assignment} and used, for example, in marriage-market models of \cite*{becker1973theory} and \cite*{galichon2022cupid}. If a couple~$(x,y)$ forms a match, then they generate a surplus~$s(x,y)$ that they may share:~$x$ getting~$ u(x,y)$ and~$y$ getting~$ v(x,y)$, with~$ u(x,y)+ v(x,y)=s(x,y)$. The model with transferable utilities assumes that the shares~$ u(x,y)$ and~$ v(x,y)$ are determined at the same time as the match. 
Effectively, the model assumes that transfers between couples are negotiated, agreed upon, and committed to, as part of the bargaining over who matches with whom.

We consider instead the possibility that agents cannot commit to a specific share of
the surplus when they agree to form a pair with another agent. Specifically, we
assume that, once a match has been formed, it cannot be broken without significant
cost, and that the two members of a couple bargain over how to share the surplus
according to the Nash bargaining model with fixed weights,~$\beta$ and~$1-\beta$, for the~$X$ and~$Y$ sides, respectively. This means
that if a couple~$(x,y)$ forms, then~$ u(x,y)=\beta \cdot s(x,y)$ while~$ v(x,y) =
(1-\beta)\cdot  s(x,y)$.

Thus a matching~$\pi$ is stable in the model with Nash bargaining
surplus-sharing if and only if it is stable in a market with aligned preferences represented by the potential $U(x,y)=s(x,y)$. The connection extends to approximate stability: $\pi$ is~$\varepsilon$-stable in the model with Nash bargaining
 if it is~$\varepsilon'$-stable in the market with aligned preferences and $\varepsilon'=\frac{\varepsilon}{\max\{\beta,\ 1-\beta\}}$.

\bibliography{main}

\appendix

\section{Proofs}\label{app_proof_equivalence}

\subsection{Extensions of Theorem~\ref{th_equivalence} and their proofs}

We prove two general results applicable to 
\begin{equation}\label{eq_general_cost_app}
c(x,y)=-h\left(U(x,y)\right),
\end{equation} 
and deduce Theorem~\ref{th_equivalence} as a corollary.
The first result addresses the case of positive~$\alpha$.
\begin{theorem}\label{th_equivalence_general}
    For a market with aligned preferences, assume that the utility function~$U$ takes values in an open interval~$I\subset \R$, possibly infinite but bounded from above. Let~$h\colon I\to (0,+\infty)$ be a differentiable function such that there exists~$\alpha>0$ satisfying
     $$\frac{h'(t)}{h(t)}\geq \alpha \quad \mbox{for all }\quad t\in I.$$
    Then any solution to the transportation problem with cost~\eqref{eq_general_cost_app} is~$\varepsilon$-stable with 
    $$\varepsilon=\frac{\ln 2}{\alpha}.$$
\end{theorem}
An example of~$h$ satisfying the requirements of Theorem~\ref{th_equivalence_general} is~$h(t)=\exp(\alpha\cdot t)$ with $\alpha>0$.

We use the following result by \cite*{beiglbock2009optimal}. Consider an optimal transportation problem
$\min_\pi \int_{X\times Y} c(x,y)\, \diff\pi(x,y)$
with Polish spaces~$X$ and~$Y$, measurable cost~$c$. Assume that the value of this transportation problem is finite and is attained at some~$\pi^*$.
\cite*{beiglbock2009optimal} establish the existence of a~$c$-cyclic monotone~$\Gamma$ such that~$\pi^*(\Gamma)=1$, i.e., the optimal transportation plan is supported on a cyclically-monotone set.
Recall that~$c$-cyclic monotonicity 
of~$\Gamma$ means that 
$$\sum_{i=1}^n c(x_i,y_i)\leq \sum_{i=1}^n c(x_i,y_{i+1})$$
for all~$n\geq 2$ and pairs of points~$(x_1,y_1),\ldots,(x_n,y_n)\in\Gamma$  with the convention~$y_{n+1} =y_1$. 
\begin{proof}[Proof of Theorem~\ref{th_equivalence_general}]
Consider now the optimal transportation problem~\eqref{eq_transport} with 
$c$ from~\eqref{eq_general_cost_app},
and let~$\pi^*$ be its solution. By the result of \cite*{beiglbock2009optimal},~$\pi^*$ is supported on some~$c$-cyclically monotone set~$\Gamma$.
The requirement of cyclic monotonicity with~$n=2$ implies that for any~$(x_1,y_1),(x_2,y_2)\in \Gamma$
$$h\big( U(x_1,y_2)\big)+h\big( U(x_2,y_1)\big)\leq h\big( U(x_1,y_1)\big)+h\big( U(x_2,y_2)\big).$$
Dropping the second term on the left-hand side and replacing both terms on the right-hand side with their maximum, we obtain
\begin{equation}\label{eq_h_cyclic_monotone_appendix}
h\big( U(x_1,y_2)\big)\leq 2 \cdot  h\Big(\max\{U(x_1,y_1), U(x_2,y_2)\}\Big).
\end{equation}
Consider a pair of points~$t,t'$ from~$I$. Integrating the bound $h'(s)/h(s)\geq \alpha$ from~$t$ to~$t'$, we get
\begin{equation}\label{eq_integrated_log_derivative}
\ln \left(\frac{h(t')}{h(t)}\right) \geq \alpha\cdot(t'-t)\quad \text{for}\quad t\leq t'.
\end{equation}
Pick~$t=\max\{U(x_1,y_1), U(x_2,y_2)\}$ and~$t'=U(x_1,y_2)$. Consider two cases depending on whether~$t\leq t'~$ or~$t> t'$. If~$t\leq t'$, plugging~$t$ and~$t'$ into~\eqref{eq_integrated_log_derivative} and taking into account the bound~\eqref{eq_h_cyclic_monotone_appendix}, we obtain
\begin{equation}\label{eq_h_bound}
\ln 2 \geq \alpha\cdot\left(U(x_1,y_2)-\max\{U(x_1,y_1), U(x_2,y_2)\}\right).
\end{equation}
If~$t>t'$, this inequality~\eqref{eq_h_bound} holds trivially as the right-hand side is negative. We conclude that 
$$U(x_1,y_2)\leq \max\{U(x_1,y_1), U(x_2,y_2)\}+\frac{\ln 2}{\alpha}$$
for a set of~$(x_1,y_1)$ and~$(x_2,y_2)$ of full~$\pi^*\times \pi^*$-measure. Thus~$\pi^*$ is~$\varepsilon$-stable with~$\varepsilon=\frac{\ln 2}{\alpha}$. 
\end{proof}

The following theorem deals with negative~$\alpha$.
\begin{theorem}\label{th_egalitarian_general}
    For a market with aligned preferences, assume that the utility function~$U$ takes values in an open interval~$I\subset \R$, possibly infinite but bounded from below. Let~$h\colon I\to (-\infty,0)$ be a differentiable increasing function  such that there exists~$\alpha<0$ satisfying
    \begin{equation*}
     \frac{h'(t)}{h(t)}\leq \alpha \quad \mbox{for all }\quad t\in I.
     \end{equation*}
    Then any solution to the transportation problem with cost~\eqref{eq_general_cost_app} is~$\varepsilon$-egalitarian with 
    $$\varepsilon=\frac{\max\{1,\, \ln |\alpha|\}}{|\alpha|}.$$
\end{theorem}
For example, $h(t)=-\exp(\alpha\cdot t)$ with $\alpha<0$ satisfies the requirements of Theorem~\ref{th_egalitarian_general}.

\begin{proof}[Proof of Theorem~\ref{th_egalitarian_general}]
Recall that~$\mathcal{U}_{\min}(\pi)$ is the essential infimum of~$U(x,y)$ with respect to measure~$\pi$, i.e.,
$\mathcal{U}_{\min}(\pi)=\inf\big\{\lambda\in \R \ \colon \  \pi(\{U(x,y)<\lambda\})>0 \big\}.$
The egalitarian lower bound is given by 
\begin{equation}\label{eq_U_star_supremum}
\mathcal{U}_{\min}^*(\mu,\nu)=\sup_{\pi \in \Pi(\mu,\nu)}\mathcal{U}_{\min}(\pi).
\end{equation}
Consider the set~$C$ of all hypothetical couples whose utility is below the egalitarian lower bound by more than~$\varepsilon$, i.e.,
$C=\left\{(x,y)\in X\times Y\colon U(x,y)< \mathcal{U}_{\min}^*(\mu,\nu)-\varepsilon \right\}.$
Let~$\pi^*$ be a solution to an optimal transportation problem with cost~\eqref{eq_general_cost_app}. Our goal is to show that~$\pi^*(C)$ cannot be too big. Fix small~$\delta>0$ and find a matching~$\pi'\in \Pi(\mu,\nu)$ such that 
$\mathcal{U}_{\min}(\pi')\geq \mathcal{U}_{\min}^*(\mu,\nu)-\delta.$
Note that we may not be able to find such a matching for~$\delta=0$ since the supremum in~\eqref{eq_U_star_supremum} may not be attained. We get
$$\int_{X\times Y} h(U(x,y))\diff\pi'(x,y)\leq \int_{X\times Y} h(U(x,y))\diff\pi^*(x,y)$$
since~$\pi^*$ is the optimal matching for the transportation problem with cost~$-h(U(x,y))$. Since~$h$ is increasing,~$h(U(x,y))\geq h(\mathcal{U}_{\min}^*(\mu,\nu)-\delta)$ for~$\pi'$-almost all pairs~$(x,y)$. Thus, the left-hand side admits the following bound
$$h(\mathcal{U}_{\min}^*(\mu,\nu)-\delta)\leq \int_{X\times Y} h(U(x,y))\diff\pi'(x,y)$$
and thus
$$h(\mathcal{U}_{\min}^*(\mu,\nu)-\delta)\leq \int_{X\times Y} h(U(x,y))\diff\pi^*(x,y).$$
Since~$h$ is continuous and~$\delta>0$ was arbitrary, we get
\begin{equation}\label{eq_bound_h(Umin)}
h(\mathcal{U}_{\min}^*(\mu,\nu))\leq \int_{X\times Y} h(U(x,y))\diff\pi^*(x,y).
\end{equation}
Using monotonicity and negativity of~$h$ and the definition of~$C$, we obtain
\begin{align*}
\int_{X\times Y} h(U(x,y))\diff\pi^*(x,y)& = \int_{C} h(U(x,y))\diff\pi^*(x,y)  + \int_{(X\times Y)\setminus C} h(U(x,y))\diff\pi^*(x,y)\\
& \leq \int_{C} h(U(x,y))\diff\pi^*(x,y) \\
& \leq h(\mathcal{U}_{\min}^*(\mu,\nu)-\varepsilon)\cdot \pi^*(C).
\end{align*}
Combining this bound with~\eqref{eq_bound_h(Umin)} gives
$$h(\mathcal{U}_{\min}^*(\mu,\nu))\leq h(\mathcal{U}_{\min}^*(\mu,\nu)-\varepsilon)\cdot \pi^*(C)$$
and thus
$$\pi^*(C)\leq \frac{h(\mathcal{U}_{\min}^*(\mu,\nu))}{h(\mathcal{U}_{\min}^*(\mu,\nu)-\varepsilon)}.$$
Note that the inequality changes direction because of the negativity of~$h$. Similarly to~\eqref{eq_integrated_log_derivative}, integrating the bound on the logarithmic derivative of~$h$, we obtain 
$$\ln \left(\frac{|h(t')|}{|h(t)|}\right) \leq \alpha\cdot(t'-t)\quad \text{for}\quad t\leq t'$$
and conclude that 
$$\pi^*(C)\leq \exp(\alpha\cdot \varepsilon)=\exp(-|\alpha|\cdot \varepsilon).$$
Plugging in~$\varepsilon=\frac{\max\{1,\, \ln |\alpha|\}}{|\alpha|}$, we get 
$$\pi^*(C)\leq \min \left\{\exp(-1),\, \frac{1}{|\alpha|}\right\}$$
and thus 
$\pi^*(C)\leq \varepsilon.$
We conclude that~$\pi^*$ is~$\varepsilon$-egalitarian with~$\varepsilon=\frac{\max\{1,\, \ln |\alpha|\}}{|\alpha|}$.
\end{proof}

Theorem~\ref{th_equivalence} follows directly from Theorems~\ref{th_equivalence_general} and~\ref{th_egalitarian_general}.
\begin{proof}[Proof of Theorem~\ref{th_equivalence}]
Consider~$h(t)=\mathrm{sign}(\alpha)\cdot\exp(\alpha\cdot t)$. For~$\alpha>0$, Theorem~\ref{th_equivalence_general} implies that the solution to the optimal transportation problem with the cost~$c(x,y)=-h\big(U(x,y)\big)$ is~$\varepsilon$-stable with~$\varepsilon=\frac{\ln 2}{\alpha}$. For~$\alpha<0$, Theorem~\ref{th_egalitarian_general} gives~$\varepsilon$-egalitarianism with~$\varepsilon=\frac{\max\{1,\, \ln |\alpha|\}}{|\alpha|}$.  
Multiplying a cost function by a positive factor and adding a constant does not affect the optimum. Thus a solution to the transportation problem with cost
$$c_\alpha(x,y)=\frac{1-\exp(\alpha\cdot U(x,y))}{\alpha}$$ has the same properties.
For $\alpha=0$, the cost is $c_0(x,y)=-U(x,y)$, so minimizing the total transport cost is exactly maximizing $W(\pi)=\int U\diff\pi$.
\end{proof}

\subsection{Proof of Theorem~\ref{th_welfare}.}
Let~$\pi$ be an~$\varepsilon$-stable matching with marginals~$\mu$ and~$\nu$. Since~$U$ is continuous, Lemma~\ref{lm_deterministic_equivalence} implies that, for any~$(x_1,y_1),(x_2,y_2)\in \supp(\pi)$, 
$$U(x_1,y_2)\leq \max\left\{U(x_1,y_1),\, U(x_2,y_2)\right\}+ \varepsilon.$$
By non-negativity of $U$, we get
$$U(x_1,y_2)\leq U(x_1,y_1)+U(x_2,y_2)+ \varepsilon.$$
Let~$\pi'$ be any other matching with marginals~$\mu$ and~$\nu$. Consider a distribution~$\lambda\in \Delta\big((X\times Y)\times (X\times Y)\big)$ such that
the marginals of~$\lambda$ on~$(x_1,y_1)$ and on~$(x_2,y_2)$ are equal to~$\pi$ and the marginal on~$(x_1,y_2)$ is~$\pi'$.\footnote{Such a $\lambda$ exists by the standard gluing construction: sample $(x_1,y_2)$ from $\pi'$, then sample $y_1$ from a conditional distribution of $\pi$ given $x_1$ and sample $x_2$ from a conditional distribution of $\pi$ given $y_2$.}
We  get
\begin{align*}
    W(\pi')&=\int_{X\times Y} U(x_1,y_2)\diff\pi'(x_1,y_2)=\int_{(X\times Y)\times (X\times Y)} U(x_1,y_2)\diff\lambda(x_1,y_1,x_2,y_2)\\
    &\leq \int_{(X\times Y)\times (X\times Y)} \left(U(x_1,y_1)+U(x_2,y_2)+ \varepsilon\right)\diff\lambda(x_1,y_1,x_2,y_2)\\
    &=\int_{X\times Y} U(x_1,y_1)\diff\pi(x_1,y_1)+ \int_{X\times Y} U(x_2,y_2)\diff\pi(x_2,y_2)+\varepsilon\\
    &=2W(\pi)+\varepsilon. 
    \end{align*}
    We thus obtain
$$W(\pi)\geq \frac{1}{2}\left(W(\pi')-\varepsilon\right)$$
     for any matching~$\pi'$. In particular, this inequality holds for~$\pi'$ maximizing welfare. Thus~$W(\pi)\geq \frac{1}{2}\left(W^*(\mu,\nu)-\varepsilon\right)$.

     Now we show that a substantial fraction of agents in an~$\varepsilon$-stable matching~$\pi$ have utilities above the egalitarian lower bound~$\mathcal{U}_{\min}^*(\mu,\nu)$. Consider the set of hypothetical couples whose utility is more than~$\varepsilon$ below~$\mathcal{U}_{\min}^*(\mu,\nu)$
$$C=\left\{(x,y)\in X\times Y\colon U(x,y)< \mathcal{U}_{\min}^*(\mu,\nu)-\varepsilon \right\}.$$
     Our goal is to show that~$\pi(C)$ cannot be too big. 
     Let~$\pi'$ be the egalitarian matching, which exists by Corollary~\ref{cor_existence}.
     Take~$\lambda$ as in the construction above for the pair~$\pi$ and~$\pi'$. In other words,~$\lambda$ is a distribution on~$(X\times Y)\times (X\times Y)$  with marginals~$\pi$ on~$(x_1,y_1)$ and~$(x_2,y_2)$ and~$\pi'$ on~$(x_1,y_2)$. 
      Thus~$U(x_1,y_2)\geq \mathcal{U}_{\min}^*(\mu,\nu)$ on a set of full~$\lambda$-measure.
     By~$\varepsilon$-stability,~$U(x_1,y_2)\leq \max\left\{U(x_1,y_1),\, U(x_2,y_2)\right\}+ \varepsilon$ and thus
$$\max\{U(x_1,y_1),\, U(x_2,y_2)\}+ \varepsilon\geq \mathcal{U}_{\min}^*(\mu,\nu).$$
     Towards a contradiction, assume that~$\pi(C)>\frac{1}{2}.$ 
     Since both marginals of~$\lambda$ are~$\pi$, the union bound gives $\lambda(C\times C)\ge 2\pi(C)-1>0$. On the other hand,
$$\max\{U(x_1,y_1),\, U(x_2,y_2)\}+ \varepsilon < \mathcal{U}_{\min}^*(\mu,\nu)$$
on~$C\times C$. This contradiction implies that~$\pi(C)\leq 1/2$ and thus any~$\varepsilon$-stable matching~$\pi$ is~$\varepsilon'$-egalitarian with~$\varepsilon'=\max\{1/2,\, \varepsilon\}$.

\subsection{Proof of Theorem~\ref{th_idiosyncractic}.}
We shall prove a more general result that applies to finite populations $X_n$ and $Y_n$ that are close to the continuous distributions $\mu$ and $\nu$, but may not be i.i.d.\  samples from these distributions. For example, $X_n$ and $Y_n$ can be given by the collections $\frac{k}{n}$-quantiles, $k=1,\ldots, n$ of $\mu$ and $\nu$. First, we formalize what is meant by ``close.'' 

\begin{definition} Probability measures $\tau$ and  $\tau'$ in $\Delta(\R^d)$ are $\varepsilon$-close if, for any $z\in \R^d$, the absolute value of the difference between their distribution functions $F(z)=\tau(\{w\in\R^d\colon w\leq z\})$ and $F'(z)=\tau'(\{w\in\R^d\colon w\leq z\})$ does not exceed $\varepsilon$.\footnote{In other words, the Kolmogorov-Smirnov distance between $\tau$ and $\tau'$ is at most $\varepsilon.$}
\end{definition}

We identify a finite collection of points $z_1,\ldots,z_n$ in  $\R^d$ with its ``empirical'' distribution $\tau_n=\frac{1}{n}\sum_{i=1}^n \delta_{z_i}$. This allows us to say that $z_1,\ldots,z_n$ is $\varepsilon$-close to a distribution $\tau\in \Delta(\R^d)$, or to another collection of points $z_1',\ldots,z_n'$. Finally, when $X_n$ and $Y_n$ are finite and of the same cardinality $n$, we say that a matching $\pi$ is deterministic if it is a matching in the usual sense. In an abuse of notation, $\Pi(X_n,Y_n)$ denotes the set of such matchings.

\begin{theorem}\label{th_idiosyncratic_appendix} Consider non-atomic $\mu,\nu\in\Delta(\R)$, let $\pi \in \Pi(\mu,\nu)$, 
and let $X_n$ and $Y_n$ be two finite populations of size $n$ that are $\varepsilon$-close to $\mu$ and $\nu$, respectively. Let $\gamma$ be a parameter such that $1\leq \gamma \leq \frac{\sqrt{n}}{2\sqrt{2}\ln n}$. Then, with probability at least 
    $$1-5 n^{1-2\gamma},$$
     there exists $\pi_n\in\Pi(X_n,Y_n)$ such that, for all matched $x_i$ and $y_j$ in $\pi_n$,  
     $$\min \left\{F_i(\xi_{i}(j)), \ G_j(\eta_{j}(i))\right\}\geq 1- 2\sqrt{2}\cdot \gamma\cdot  \frac{\ln n}{\sqrt{n}}$$ and 
     $\pi_n$ is $\varepsilon'$-close to $\pi$
     with 
     $$\varepsilon'=9\max\left\{\varepsilon, \ \sqrt{\frac{\ln (4(n+1))}{2n}} \right\}+\frac{1}{\gamma \ln n}.$$
\end{theorem}

Recall the Dvoretzky–Kiefer–Wolfowitz inequality \citep{dvoretzky1956asymptotic}, which states that the empirical distribution of a sample from a distribution is close to the distribution with high probability. An independent sample $z_1,\ldots,z_n$ from a distribution $\tau\in\Delta(\R^d)$ is $\varepsilon$-close to $\tau$ with probability at least 
$1-2\cdot \exp\left(-2n\varepsilon^2\right)$ for $d=1$ \citep{massart1990tight} and  $1-2d(n+1)\cdot \exp\left(-2n\varepsilon^2\right)$ for $d\geq2$  \citep{naaman2021tight}. 

By the Dvoretzky–Kiefer–Wolfowitz inequality, if $X_n$ and $Y_n$ are i.i.d.\ samples from $\mu$ and $\nu$, then, for large $n$, these samples become arbitrarily close to $\mu$ and $\nu$ with high probability. Theorem~\ref{th_idiosyncractic} follows from Theorem~\ref{th_idiosyncratic_appendix} applied to such i.i.d.\ samples, with $\delta_n$ equal to the maximum of all the approximation errors.

We prove Theorem~\ref{th_idiosyncratic_appendix} in two steps. We first show that any continuous matching $\pi\in \Pi(\mu,\nu)$ can be approximated by a matching of finite populations $X_n$ and $Y_n$, without providing any guarantee on agents' utilities with respect to their idiosyncratic components. Second, we demonstrate that any matching of $X_n$
and $Y_n$, with high probability, can be modified so that the new matching is close to the original one and each agent is close to getting their best utilities with respect to the idiosyncratic component.

\begin{proposition}\label{prop:approximation}
    Let $\pi\in \Pi(\mu,\nu)$ with non-atomic $\mu,\nu\in \Delta(\R)$. Let $X_n$ and $Y_n$ be~$\varepsilon$-close to $\mu$ and $\nu$, respectively. Then there exists a deterministic matching $\pi_n$ of $X_n$ and $Y_n$ that is $\delta$-close to $\pi$ with 
    $$\delta=9\max\left\{\varepsilon, \ \sqrt{\frac{\ln (4(n+1))}{2n}} \right\}.$$
\end{proposition}
The first step in proving Proposition~\ref{prop:approximation} is to show that any matching $\pi$ can be approximated with an auxiliary deterministic matching of \emph{some} finite populations $X_n'$ and $Y_n'$, i.e., in contrast to Proposition~\ref{prop:approximation} the finite populations are not given.
\begin{lemma}\label{lem:approximation_free_marginals}
Let $\pi\in \Pi(\mu,\nu)$ with\footnote{Note that this lemma does not require $\mu$ and $\nu$ to be non-atomic. However, for atomic $\mu$ and $\nu$, we are not guaranteed that $X_n'$ and  $Y_n'$ contain $n$ \emph{distinct} points.} $\mu,\nu\in \Delta(\R)$. Then there exist finite populations $X_n'$ and $Y_n'$ of size $n$ and  $\pi_n'\in\Pi(X'_n,Y'_n)$ such that $\pi_n'$ is ${\varepsilon'}$-close to $\pi$, $X_n'$ is ${\varepsilon'}$-close to $\mu$, and $Y_n'$ is ${\varepsilon'}$-close to $\nu$ with
$${\varepsilon'}=\sqrt{\frac{\ln (4(n+1))}{2n}}.$$
\end{lemma}
\begin{proof}
Sample $n$ points $z_1',\ldots,z_n'$ independently from $\pi$. By the Dvoretzky–Kiefer–Wolfowitz inequality, the sample $z_1',\ldots,z_n'$ is ${\varepsilon'}$-close to $\pi$ with probability at least 
$1-4(n+1)\cdot \exp\left(-2n{\varepsilon'}^2\right)$. This probability is positive for any ${\varepsilon'}>\sqrt{\frac{\ln (4(n+1))}{2n}}$. 
Thus, for any such ${\varepsilon'}$, there exists a collection of $n$ points that is ${\varepsilon'}$-close to $\pi$. 
Letting $\varepsilon'$ go to $\sqrt{\frac{\ln (4(n+1))}{2n}}$ from above, and choosing a convergent subsequence, there exists such  $z_1',\ldots,z_n'$ for ${\varepsilon'}=\sqrt{\frac{\ln (4(n+1))}{2n}}$ as well.\footnote{Closeness is invariant under monotone reparameterization, so we may assume that $\mu$ and $\nu$ are supported on $[0,1]$, which justifies the compactness argument.}
Each point $z_i'$ is a pair $(x_i',y_i')$. Since $z_1',\ldots,z_n'$ is ${\varepsilon'}$-close to $\pi$, the marginal $X_n'=\{x_1',\ldots,x_n'\}$ is ${\varepsilon'}$-close to $\mu$ and $Y_n'=\{y_1',\ldots,y_n'\}$ is ${\varepsilon'}$-close to $\nu$. We demonstrate this for $X_n'$; the argument for $Y_n'$ is analogous. Define the empirical distribution of $X_n'$ as $\mu_n'=\frac{1}{n}\sum_{i=1}^n \delta_{x_i'}$. Since $\pi_n'$ is ${\varepsilon'}$-close to $\pi$, we have that
$\left|\pi(\{(x,y)\colon (x,y)\leq (t_x,t_y)\})-\pi_n'(\{(x,y)\colon (x,y)\leq (t_x,t_y)\})\right|\leq {\varepsilon'}$ for any $(t_x,t_y)$. This implies
$$\left|\mu(\{x\leq t_x\})-\mu_n'(\{x\leq t_x\})\right|\leq {\varepsilon'}$$
and thus $X_n'$ is ${\varepsilon'}$-close to $\mu$.
\end{proof}
To prove Proposition~\ref{prop:approximation}, without loss of generality, we can assume that $\mu$ and $\nu$ are uniform on $[0,1]$, which is equivalent to working in the space of quantiles for general $\mu$ and $\nu$. The next lemma shows that if a finite population $w_1,\ldots, w_n$ is close to the uniform distribution $[0,1]$, then the points are close to the equidistant points $i/n$.
\begin{lemma}\label{lem:close_uniform}
Let $w_1,\ldots, w_n$ be numbers in $[0,1]$ ordered so that $w_i\leq w_{i+1}$. If they are $\varepsilon$-close to the uniform distribution on $[0,1]$, then for any $i=1,\ldots,n$ we have
$$\left|w_i-\frac{i}{n}\right|\leq \varepsilon.$$
\end{lemma}
\begin{proof}
Let $\mu_n=\frac{1}{n}\sum_{k=1}^n \delta_{w_k}$ and let $\mu$ be uniform on $[0,1]$, so $\mu([0,t])=t$ for $t\in[0,1]$.
Since $\mu_n$ is $\varepsilon$-close to $\mu$, we have $|\mu([0,t])-\mu_n([0,t])|\le\varepsilon$ for all $t$. We fix $i$ and plug in $t=w_i+\delta$ and $t=w_i-\delta$ with $\delta>0.$
Let $\underline{i}=\min\{k:w_k=w_i\}$ and $\overline{i}=\max\{k:w_k=w_i\}$.
For all sufficiently small $\delta>0$ we have $\mu_n([0,w_i-\delta])=(\underline{i}-1)/n$ and $\mu_n([0,w_i+\delta])=\overline{i}/n$. Letting $\delta\to 0$ gives
$$\left|w_i-\frac{\underline{i}-1}{n}\right|\le \varepsilon \qquad\text{and}\qquad \left|w_i-\frac{\overline{i}}{n}\right|\le \varepsilon.$$
Since $\underline{i}\le i\le \overline{i}$, we obtain $|w_i-i/n|\le\varepsilon$.
\end{proof}
The following result allows us to show that a deterministic matching of $X_n$ and $Y_n$ and a deterministic matching of $X_n'$ and $Y_n'$ are close to each other if matched couples are close.
\begin{lemma}\label{lem:close_matchingsclose}
Let $X_n, Y_n, X_n'$, and $Y_n'$ be finite collections of $n$ points in $[0,1]$ in non-decreasing order. Suppose $X_n, X_n', Y_n,$  and $Y_n'$ are $\alpha$-close to the uniform distribution on $[0,1]$. 
If $\pi_n\in\Pi(X_n,Y_n)$ and $\pi'_n\in\Pi(X'_n,Y'_n)$ are such that $x_i$ is matched with $y_j$ if and only if $x_i'$ is matched with $y_j'$, then $\pi_n$ and $\pi_n'$ are $8\alpha$-close to each other.
\end{lemma}
\begin{proof}
By Lemma~\ref{lem:close_uniform}, we have that $|x_i-\frac{i}{n}|\le \alpha$ and $|x_i'-\frac{i}{n}|\le \alpha$, and similarly
$|y_j-\frac{j}{n}|\le \alpha$ and $|y_j'-\frac{j}{n}|\le \alpha$. Hence $|x_i-x_i'|\le 2\alpha$ and $|y_j-y_j'|\le 2\alpha$.

Fix $z=(z_1,z_2)\in [0,1]^2$ and let $z^+=(z_1^+,z_2^+)$ where $z_k^+=\min\{z_k+2\alpha,1\}$.
If $(x_i,y_j)\le z$ is a matched pair under $\pi_n$, then $(x_i',y_j')$ is the corresponding matched pair under $\pi_n'$ and satisfies
$(x_i',y_j')\le z^+$. Therefore
$$
\pi_n\left(\{w\in [0,1]^2\colon w\le z \}\right)\leq \pi_n'\left(\{w\in [0,1]^2\colon w\le z^+ \}\right),
$$
and similarly,
$$
\pi_n'\left(\{w\in [0,1]^2\colon w\le z \}\right)\leq \pi_n\left(\{w\in [0,1]^2\colon w\le z^+ \}\right).
$$

Denote $\{z<w\le z^+\}=\{w\in[0,1]^2:\ w\le z^+\ \text{and}\ w\nleq z\}$ and  let $\mu_n,\nu_n$ be the marginals of $\pi_n$. We obtain
$$
\pi_n\left(\{z<w\le z^+\}\right)\le \mu_n((z_1,z_1^+])+\nu_n((z_2,z_2^+]).
$$
Because $\mu_n$ is $\alpha$-close to the uniform distribution on $[0,1]$, for any $t\in[0,1]$ we have
$\mu_n([0,t])\le t+\alpha$ and $\mu_n([0,t])\ge t-\alpha$, hence
$\mu_n((z_1,z_1^+])=\mu_n([0,z_1^+])-\mu_n([0,z_1])\le (z_1^+-z_1)+2\alpha\le 4\alpha$,
where the last inequality uses $z_1^+-z_1\le 2\alpha$. The analogous bound holds for $\nu_n((z_2,z_2^+])$, and thus
$$\pi_n(\{z<w\le z^+\})\le 8\alpha.$$ The same argument applied to $\pi_n'$ gives $\pi_n'(\{z<w\le z^+\})\le 8\alpha$.  Thus
$$
\left|\pi_n\left(\{w\in [0,1]^2\colon w\le z \}\right) - \pi_n'\left(\{w\in [0,1]^2\colon w\le z \}\right)\right|\le 8\alpha
$$
for any $z$, and thus $\pi_n$ and $\pi_n'$ are $8\alpha$-close to each other.
\end{proof}
We are now ready to prove Proposition~\ref{prop:approximation}.
\begin{proof}[Proof of Proposition~\ref{prop:approximation}]
As discussed above, we can assume that $\mu$ and $\nu$ are uniform on $[0,1]$ by working in the space of quantiles. Let $X'_n$, $Y'_n$, and $\pi'_n\in\Pi(X'_n,Y'_n)$ be obtained from  Lemma~\ref{lem:approximation_free_marginals}.
Suppose points in $X_n,X_n',Y_n,$ and $Y_n'$ are all ordered in non-decreasing order. 
Define a matching $\pi_n$ of $X_n$ and $Y_n$ as follows: $x_i$ and $y_j$ are matched if and only if $x_i'$ and $y_j'$ are matched under $\pi_n'$. Taking $\alpha=\max\{\varepsilon, \varepsilon'\}$, we have that $X_n$ and $X_n'$ are $\alpha$-close to the uniform distribution on $[0,1]$ and $Y_n$ and $Y_n'$ are $\alpha$-close to the uniform distribution on $[0,1]$. By Lemma~\ref{lem:close_matchingsclose}, we have that $\pi_n$ and $\pi_n'$ are $8\alpha$-close to each other and thus $\pi_n$ is $9\alpha$-close to $\pi$, completing the proof.
\end{proof}

The next step in proving Theorem~\ref{th_idiosyncratic_appendix} is to show that any matching of $X_n$ and $Y_n$ can be modified so that the new matching is close to the original one and each agent is close to getting their best utilities with respect to the idiosyncratic component.
\begin{proposition}\label{prop:high_idiosyncractic}
    Let $\pi_n$ be a deterministic matching of two finite populations $X_n$ and $Y_n$ of size $n\geq 241$, and let $\xi_{i}(j)$ and $\eta_{j}(i)$ be independent random variables with continuous distribution functions $F_i$ and $G_j$, respectively. For any $\gamma$ such that $1\leq \gamma \leq \frac{\sqrt{n}}{2\sqrt{2}\ln n}$, with probability at least 
    $$1-5 n^{1-2\gamma}$$
     there exists $\pi'_n\in\Pi(X_n,Y_n)$ such that for all matched $x_i$ and $y_j$ we have 
     $$\min \left\{F_i(\xi_{i}(j)), \ \  G_j(\eta_{j}(i))\right\}\geq 1- \delta,\qquad \text{where}\qquad \delta=2\sqrt{2}\cdot \gamma\cdot  \frac{\ln n}{\sqrt{n}}$$ and 
     $\pi_n'$ is $\varepsilon$-close to $\pi_n$
     with 
     $$\varepsilon=\frac{1}{\gamma \ln n}.$$
\end{proposition}
Note that the requirement $n\geq 241$ is equivalent to requiring that the range $1\leq \gamma \leq \frac{\sqrt{n}}{2\sqrt{2}\ln n}$ is non-empty. To prove the proposition,
we will need the following lemma resembling Lemma~\ref{lem:close_matchingsclose}.
\begin{lemma}\label{lem:close_matchingsclose2}
Let $X_n=\{x_1,\ldots, x_n\}$ and $Y_n=\{y_1,\ldots, y_n\}$ be two collections of $n$ distinct numbers in non-decreasing order. Let $\pi_n,\pi'_n\in\Pi(X_n,Y_n)$, and $t$ be a positive integer such that, if $x_i$ is matched with $y_j$ under $\pi_n$ and to $y_k$ under $\pi_n'$, then $|j-k|\leq t-1$. Then $\pi_n$ and $\pi_n'$ are $\frac{t-1}{n}$-close to each other.
\end{lemma}
\begin{proof}
Applying a monotone reparameterization to $X_n$ and $Y_n$, we can assume $x_i=y_i=i/n$ without loss of generality. Denote $\delta=(t-1)/n$ and fix $z=(z_1,z_2)\in [0,1]^2$ and let $z^+=(z_1,z_2+\delta)$. We have
$$\pi_n\left(\{w\in [0,1]^2\colon w\leq z \}\right)\leq \pi_n'\left(\{w\in [0,1]^2\colon w\leq z^+ \}\right)$$
and similarly
$$\pi_n'\left(\{w\in [0,1]^2\colon w\leq z \}\right)\leq \pi_n\left(\{w\in [0,1]^2\colon w\leq z^+ \}\right).$$
Note that
$$\pi_n\left(\{w\in [0,1]^2\colon z< w\leq z^+ \}\right)\leq \pi_n\left(\{w\in [0,1]^2\colon z_2< w_2\leq z_2+\delta \}\right).$$
The right-hand side is equal to the fraction of $i=1,\ldots n$ such that $i/n\in (z_2,z_2+\delta]$. Thus
$$\pi_n\left(\{w\in [0,1]^2\colon z< w\leq z^+ \}\right)\leq \delta.$$
Analogous calculations for $\pi_n'$ yield the same bound, and thus 
$$\left|\pi_n\left(\{w\in [0,1]^2\colon w\leq z \}\right) - \pi_n'\left(\{w\in [0,1]^2\colon w\leq z \}\right)\right|\leq \delta.$$
Since $z$ was arbitrary, we have that $\pi_n$ and $\pi_n'$ are $\delta$-close to each other.
\end{proof}

We are now ready to prove Proposition~\ref{prop:high_idiosyncractic}. The proof adapts the approach of \cite{erdHos1964random}, who studied the existence of a perfect matching in a random bipartite graph where each edge is included independently with probability $p$, to graphs where some edges are never traced. Leveraging insights of \cite{Petrov}, we obtain the explicit bound that applies to any fixed~$n$ rather than the asymptotic result for~$n\to\infty$ as in the classical paper.

\begin{proof}[Proof of Proposition~\ref{prop:high_idiosyncractic}]
Suppose $X_n$ and $Y_n$ are ordered in non-decreasing order.
Consider a bipartite graph $\Gamma=(X_n,Y_n,E)$ with vertices $X_n$ and $Y_n$ obtained by the edges in $\pi_n$ and additional edges. Specifically, $E$ contains an edge  between all $x_i$ and $y_k$ such that $x_i$ is matched with some $y_j$ in $\pi_n$ and $|k-j|\leq t-1$.  Here~$t$ is  
$$
t=\left\lceil\frac{\sqrt{{(2\ln 4)\cdot n}}}{\delta}\right\rceil,
$$
where $\delta=2\sqrt{2}\cdot \gamma\cdot  \frac{\ln n}{\sqrt{n}}$, and $\lceil z\rceil$ denotes the ceiling of $z$. 

Consider a random subgraph $\Gamma_\delta=(X_n,Y_n, E_{q})$ of $\Gamma$ where only those
edges $e=(i,k)\in E$ are left for which $F_i(\xi_{i}(k))\geq 1-\delta$ and $G_k(\eta_{k}(i))\geq 1-\delta$. In other words, $\Gamma_\delta$ is a subgraph of $\Gamma$ where each edge is eliminated with probability $q=1-\delta^2$. We aim to show 
that, with high probability, $\Gamma_\delta$ contains a perfect matching. Conditional on
this high-probability event, we define a matching $\pi_n'$ of $X_n$ and $Y_n$ as
follows: $x_i$ is matched with $y_l$ if $e=(x_i,y_l)$ enters the perfect matching of $\Gamma_\delta$. By Lemma~\ref{lem:close_matchingsclose2}, $\pi_n$ and~$\pi_n'$ are $\frac{t-1}{n}$-close to each other. Since $t=\left\lceil\frac{\sqrt{{(2\ln 4)\cdot n}}}{\delta}\right\rceil$, we have $t-1<\frac{\sqrt{{(2\ln 4)\cdot n}}}{\delta}$ and thus
$$
\frac{t-1}{n}<\frac{\sqrt{{(2\ln 4)\cdot n}}}{\delta\,n}
=\frac{\sqrt{\ln 4}}{2\gamma \ln n}\leq \frac{1}{\gamma \ln n}.
$$
We conclude that $\pi_n$ and~$\pi_n'$ are $\varepsilon$-close with 
$$
\varepsilon=\frac{1}{\gamma \ln n}.
$$
By the definition of $\Gamma_\delta,$ $F_i(\xi_{i}(l))\geq 1-\delta$ and $G_l(\eta_{l}(i))\geq 1-\delta$ for any matched pair~$(x_i,y_l).$ 

To complete the proof, we must bound the probability that $\Gamma_\delta$ admits a perfect matching. By Hall's theorem, $\Gamma_\delta$ contains a perfect matching if and only if for any subset $S\subset X_n$, the number of neighbors of $S$ in $Y_n$ is at least $|S|$. Equivalently, for any $S\subset X_n$ and $T\subset Y_n$ such that $|S|+|T|=n+1$, there must be an edge between $S$ and $T$ in $\Gamma_\delta$. Consider $S$ and $T$ such that $|S|\leq |T|$ and let $m=|S|$. For $t\leq n/2$,  the number of edges $|E_G(S,T)|$ between such $S$ and $T$ in $\Gamma$ is at least
\begin{equation}\label{eq_bound_on_E(S,T)}
    |E_G(S,T)|\geq 
\begin{cases}       
    m\cdot \frac{2t-m+1}{2}, & \text{if } m\leq t,\\
    t\cdot \frac{t+1}{2}, & \text{if } m> t
\end{cases}
.
\end{equation}
This lower bound corresponds to the ``diagonal'' matching $\pi_n$---i.e., $x_i$ is matched with $y_i$---and sets  $S=\{x_1,\ldots, x_m\}$ and $T=\{y_m,y_{m+1},\ldots, y_n\}$. By definition of $\Gamma$, there is an edge between $x_i$ and each $y_k$ with  index $k$ within $t-1$ from $x_i$'s match. Thus for $m\leq t$, there are edges between each $x_i\in S$ and $y_m,\ldots, y_{t+i-1}$. For $m>t$, there are edges between $x_i\in S$ with $i\geq m-t+1$ and  $y_m,\ldots, y_{t+i-1}$. Counting the numbers of such edges, we obtain~\eqref{eq_bound_on_E(S,T)}.

Each edge of $\Gamma$ is eliminated in $\Gamma_\delta$ with probability $q$. Thus there are no edges between $S$ and $T$ in $\Gamma_\delta$ with probability $q^{|E_G(S,T)|}$, and the union bound applied to the probability that $\Gamma_\delta$ does not contain a perfect matching gives the following
\begin{align*}
\P\left(\text{$\Gamma_\delta$ has no perfect matching}\right)\leq \sum_{m=1}^n \ \ \ 
\sum_{S,T\colon 
    |S|=m,\  |T|=n-m+1}
q^{|E_G(S,T)|}.
\end{align*}
Denote the right-hand side of~\eqref{eq_bound_on_E(S,T)} by $E^*(m,t)$. Thus the number of edges between $S$ and $T$ in $\Gamma$ is at least $E^*(m,t)$ if $m=|S|\leq |T|=n-m+1$. Since the roles of $S$ and $T$ are symmetric, the number of edges is at least $E^*(n+1-m,t)$ for $m\geq n-m+1$. 
Taking into account that the number of subsets $S$ and $T$ with $|S|=m$ is given by $\binom{n}{m}\binom{n}{n+1-m}$,
we obtain
\begin{align*}
\P\left(\text{$\Gamma_\delta$ has no perfect matching}\right)
&\leq \sum_{m\leq \frac{n+1}{2}}\binom{n}{m}\binom{n}{n+1-m}q^{E^*(m,t)}\\
&+ \sum_{m\geq \frac{n+1}{2}}\binom{n}{m}\binom{n}{n+1-m}q^{E^*(n+1-m,t)}\\
&\leq 2 \sum_{m\leq \frac{n+1}{2}}\binom{n}{m}\binom{n}{n+1-m}q^{E^*(m,t)}.
\end{align*}
Plugging in the explicit value of $E^*(m,t)$, we get
\begin{align}\label{eq_the_two_sums}
\P\left(\text{$\Gamma_\delta$ has no perfect matching}\right)
&\leq 2\sum_{m=1}^t \binom{n}{m}\binom{n}{n+1-m}q^{m\cdot \frac{2t-m+1}{2}}\\
&+ 2\sum _{t+1\leq m\leq \frac{n+1}{2}} \binom{n}{m}\binom{n}{n+1-m}q^{ t\cdot \frac{t+1}{2}}.\notag
\end{align}
Consider the first sum in~\eqref{eq_the_two_sums}. We have $m\cdot \frac{2t-m+1}{2}\geq m\cdot \frac{t}{2}$. Using a rough bound $\binom{n}{m}\leq n^m$ and
$\binom{n}{n+1-m}=\binom{n}{m-1}\leq n^{m-1}$, and assuming $n^2\cdot q^\frac{t}{2}<1$, we get
\begin{align}\notag
\sum_{m=1}^t \binom{n}{m}\binom{n}{n+1-m}q^{ m\cdot \frac{2t-m+1}{2}}
&\leq \sum_{m=1}^t n^{2m-1}q^{m\cdot \frac{t}{2}}\\
&\leq n\cdot q^\frac{t}{2}\sum_{l=0}^\infty \left(n^2\cdot q^\frac{t}{2}\right)^l
= \frac{nq^\frac{t}{2}}{1-n^2\cdot q^\frac{t}{2}}.\label{eq_the_1st_sum_simplified}
\end{align}
We now ensure that $n^2\cdot q^\frac{t}{2}<1$, and thus the infinite geometric series converge. Indeed, 
\begin{align*}
n^2 q^\frac{t}{2}
&= n^2 (1-\delta^2)^\frac{t}{2}\leq n^2 \left(\exp(-\delta^2)\right)^{\frac{t}{2}}
= \exp\left(2\ln n -\delta^2\cdot \frac{t}{2} \right),
\end{align*}
where we used the fact that $1+z\leq \exp(z)$ for any $z\in \mathbb{R}$. 
Since $t \geq \frac{\sqrt{{(2\ln 4)\cdot n}}}{\delta}$ and $\delta=2\sqrt{2}\cdot \gamma\cdot  \frac{\ln n}{\sqrt{n}}$, we get 
$\delta^2\cdot \frac{t}{2} \geq 2 \sqrt{\ln 4}\cdot \gamma\cdot \ln n$ and thus 
\begin{equation}\label{eq_n^2q^t/2_bound}
n^2 q^{\frac{t}{2}}
\leq \exp\left(2\ln n -\delta^2\cdot \frac{t}{2} \right)
\leq n^{2(1-\gamma\cdot \sqrt{\ln 4})}.
\end{equation}
For $n=241$ and $\gamma=1$, we get $n^{2(1-\gamma\cdot \sqrt{\ln 4})}\approx 0.143< 1/2$. Therefore, for any $n\geq 241$ (and $\gamma\ge 1$), we obtain $n^2 q^\frac{t}{2}\leq 1/2$, and thus the infinite geometric series converges.
Moreover, since $nq^\frac{t}{2}= \frac{1}{n}\cdot n^2 q^\frac{t}{2}$, we get
\begin{equation}\label{eq_nq^t/2_bound}
nq^\frac{t}{2}\leq  n^{1-2\gamma}
\end{equation}
from~\eqref{eq_n^2q^t/2_bound}, where we replaced $\sqrt{\ln 4}$ with~$1$. 
Thus the first sum in~\eqref{eq_the_two_sums} can be bounded as follows:
$$
\sum_{m=1}^t \binom{n}{m}\binom{n}{n+1-m}q^{m\cdot \frac{2t-m+1}{2}}
\leq 2 n^{1-2\gamma},
$$ 
where we combined the bound~\eqref{eq_the_1st_sum_simplified} with~\eqref{eq_nq^t/2_bound} and $n^2 q^\frac{t}{2}\leq 1/2$.

Consider now the second sum in~\eqref{eq_the_two_sums}. The product of binomial coefficients is maximized at $m=\lfloor n/2\rfloor$:
$$
\binom{n}{m}\binom{n}{n+1-m}\leq \binom{n}{\lfloor n/2\rfloor}\cdot \binom{n}{\lceil n/2\rceil}
=\binom{n}{\lfloor n/2\rfloor}^2.
$$
Applying the Stirling formula, the second sum admits the following bound
\begin{align*}
\sum _{t+1\leq m\leq \frac{n+1}{2}} \binom{n}{m}\binom{n}{n+1-m}q^{t\cdot \frac{t+1}{2}}
&\leq \sum _{t+1\leq m\leq \frac{n+1}{2}} \binom{n}{\lfloor n/2\rfloor}^2 q^{t\cdot \frac{t+1}{2}}\\
&\leq \frac{n}{2}\cdot \frac{2\cdot 2^{2n}}{\pi n}q^{t\cdot \frac{t+1}{2}}.
\end{align*}
Since $t \geq \frac{\sqrt{{(2\ln 4)\cdot n}}}{\delta}$, we have that 
$$
2^{2n}\cdot q^{\frac{t^2}{2}}
= 2^{2n}\cdot(1-\delta^2)^{\frac{t^2}{2}}
\leq 2^{2n}\cdot\exp\left(-\delta^2\cdot\frac{t^2}{2}\right)
\leq 2^{2n}\cdot e^{-n\cdot \ln 4}= 1.
$$
Thus
$$
\frac{n}{2}\cdot \frac{2\cdot 2^{2n}}{\pi n}q^{t\cdot \frac{t+1}{2}}
\leq \frac{1}{\pi}\cdot q^\frac{t}{2}.
$$
By~\eqref{eq_nq^t/2_bound}, $q^\frac{t}{2}\leq  n^{-2\gamma}$, and thus
$$
\sum _{t+1\leq m\leq \frac{n+1}{2}} \binom{n}{m}\binom{n}{n+1-m}q^{t\cdot \frac{t+1}{2}}
\leq \frac{1}{\pi }\cdot n^{-2\gamma}.
$$
Putting the bounds for the first and the second sum together, we obtain that
$$
\P\left(\text{$\Gamma_\delta$ has no perfect matching}\right)
\leq \left(4+\frac{2}{\pi\cdot n}\right) n^{1-2\gamma}
\leq 5\cdot  n^{1-2\gamma},
$$
completing the proof.
\end{proof}

\begin{proof}[Proof of Theorem~\ref{th_idiosyncratic_appendix}]
By Proposition~\ref{prop:approximation}, there is a deterministic matching of $X_n$ and $Y_n$ that is $9\max\{\varepsilon,\sqrt{\ln(4(n+1))/(2n)}\}$-close to $\pi$. Proposition~\ref{prop:high_idiosyncractic} modifies this matching, with probability at least $1-5n^{1-2\gamma}$, to one in which every matched pair attains idiosyncratic quantiles at least $1-2\sqrt{2}\gamma\ln n/\sqrt n$, while changing the matching by at most $1/(\gamma\ln n)$. The triangle inequality for the Kolmogorov-Smirnov distance gives the stated bound on $\varepsilon'$.
\end{proof}

\clearpage
\thispagestyle{empty}
\vfill

\begin{center}
   \textbf{ {\Large Supplementary Appendix }}
\end{center}
\vfill

\clearpage
\setcounter{page}{1}

\section{Stability in markets with a continuous potential}\label{app_lemma_pointwise_stability}

\begin{lemma}\label{lm_deterministic_equivalence}
For a market with a continuous potential~$U$, a matching~$\pi$ is~$\varepsilon$-stable if and only if, for all~$(x_1,y_1),(x_2,y_2)\in\supp(\pi)$, the $\varepsilon$-stability condition~\eqref{eq_epsilon_stable_common} holds.
\end{lemma}
\begin{proof}
One direction is straightforward: the pointwise property implies the almost-everywhere property. We prove the opposite direction.
Let~$\pi$ be an~$\varepsilon$-stable matching with continuous utilities~$U$. Our goal is to show that,  for any~$(x_1,y_1),(x_2,y_2)\in \supp(\pi)$, the inequality~\eqref{eq_epsilon_stable_common} holds. Towards a contradiction, suppose that there are~$(x_1^*,y_1^*),(x_2^*,y_2^*)\in \supp(\pi)$ such that this inequality is violated. By continuity of~$U$, we can find open neighborhoods~$V_1\subseteq X\times Y$ of~$(x_1^*,y_1^*)$ and~$V_2\subseteq X\times Y$ of~$(x_2^*,y_2^*)$ such that the inequality is violated for all~$(x_1,y_1),(x_2,y_2)\in V_1\times V_2$. Since~$(x_1^*,y_1^*),(x_2^*,y_2^*)\in \supp(\pi)$, we have~$\pi(V_1)>0$ and~$\pi(V_2)>0$. Thus~$V_1\times V_2$ has a positive~$\pi\times \pi$-measure, which contradicts the~$\varepsilon$-stability of~$\pi$. 
We conclude that, for a continuous potential,~$\varepsilon$-stability of~$\pi$ implies that~\eqref{eq_epsilon_stable_common} holds pointwise on the support of~$\pi$.
\end{proof}

\section{Exact algorithm for stable matching on \texorpdfstring{$\R$}{R}}\label{sec_algorithm}
When~$X$ and~$Y$ are finite, a simple greedy algorithm will result in a stable matching: let~$(x,y)$ be the pair with the highest utility from matching. Match~$(x,y)$, remove the matched agents, and repeat. In this section, we construct an analogous procedure for the case where~$X$ and~$Y$ are possibly continuous. 
\medskip

We denote by $\mathcal M_+(Z)$ the set of finite nonnegative measures on $Z$.
Let $X=Y=\mathbb R$ and consider two populations represented by measures
$\mu\in\mathcal M_+(X)$ and $\nu\in\mathcal M_+(Y)$ with the same total mass.
Agents of types $x\in X$ and $y\in Y$ have utility $U(x,y)=-|x-y|$.
We refer to the pair $(\mu,\nu)$ as a \emph{market}.
In this section we do not normalize the total mass to $1$, since our construction
will repeatedly remove and match submarkets. The notions of a matching and of
stability extend in a straightforward way to any $\mu,\nu$ with $\mu(X)=\nu(Y)$.

We assume that $\mu$ and $\nu$ are non-atomic and that the signed measure
$\rho=\mu-\nu$ changes sign only finitely many times. We say that $\rho$
changes sign $K$ times if there exists a collection of intervals
$I_0,\ldots,I_K$ such that $\rho$ puts no mass outside of their union and, on each $I_k$, the restriction of $\rho$ is
either a nonnegative measure or a nonpositive measure, and $K$ is the minimal non-negative integer with
this property. We refer to the intervals $I_k$ as the \emph{imbalance regions}.

\begin{proposition}\label{th_algorithm_appendix}
There is a unique stable matching $\pi$, and it can be constructed via an
algorithm. Moreover, for each $x\in X$ there are at most two distinct types
$y,y'\in Y$ such that both $(x,y)$ and $(x,y')$ belong to $\supp(\pi)$.
If $\mu$ and $\nu$ have piecewise-constant densities with at most $m$ intervals
of constancy, the algorithm runs in time on the order of $m^2$.
\end{proposition}

A stable matching for $(\mu,\nu)$ is constructed by sequentially simplifying the
market. The idea is to identify submarkets $(\mu',\nu')$ that must be matched
together in any stable matching in exactly the same way as they would be
matched if the remainder of the population $(\mu-\mu',\nu-\nu')$ did not exist.
As we will see, one can always choose such a submarket so that it is matched in
a simple monotone manner. Iterating this procedure yields the stable matching
for $(\mu,\nu)$. The number of steps is finite because the number of imbalance
regions is finite and decreases monotonically under the sequential simplification of the market.
For piecewise-constant densities with $m$ intervals of constancy, the number of
sign changes is at most $m$, and each simplification step can be implemented
using $O(m)$ operations, giving overall quadratic runtime.

\medskip

We next revisit the link between stability and the no-crossing property
discussed in Section~\ref{sec:applications}. There we showed that a stable
matching satisfying no-crossing exists. We now show that, in fact, {every}
stable matching must satisfy this property.

For $z_1,z_2\in\mathbb R$, let $O(z_1,z_2)$ be the smallest circle in
$\mathbb R^2$ containing the points $(z_1,0)$ and $(z_2,0)$ (equivalently, the
circle with diameter endpoints $(z_1,0)$ and $(z_2,0)$). Consider a matching
$\pi$ on $\mathbb R$. We say that $\pi$ satisfies {no-crossing} if, for any
two pairs $(x,y)$ and $(x',y')$ in $\supp(\pi)$, the circles $O(x,y)$ and
$O(x',y')$ do not intersect unless $x=x'$ or $y=y'$.

\begin{figure}[ht]
    \centering
    \begin{subfigure}{0.31\textwidth}
    \centering
    \begin{tikzpicture}[scale=0.4]
        \draw[black,very thick] (0,0) -- (10.5,0);
        \draw[very thick, dashed, ->] (1,0) arc (180:0:2);
        \draw[very thick, dashed, ->] (3.5,0) arc (180:0:3);
        \node at (1,-0.5) {\footnotesize~$x$};
        \node at (5,-0.5) {\footnotesize~$y$};
        \node at (3.5,-0.5) {\footnotesize~$x'$};
        \node at (9.5,-0.5) {\footnotesize~$y'$};
    \end{tikzpicture}
        \caption{blocked by $(x',y)$}
            \label{fig:blocking1}
    \end{subfigure}
  \begin{subfigure}{0.31\textwidth}
    \centering
    \begin{tikzpicture}[scale=0.4]
        \draw[black,very thick] (0,0) -- (10.5,0);
        \draw[very thick, dashed, ->] (1,0) arc (180:0:2);
        \draw[very thick, dashed, ->] (9.5,0) arc (0:180:3);
        \node at (1,-0.5) {\footnotesize~$x$};
        \node at (5,-0.5) {\footnotesize~$y$};
        \node at (3.5,-0.5) {\footnotesize~$y'$};
        \node at (9.5,-0.5) {\footnotesize~$x'$};
    \end{tikzpicture}
        \caption{blocked by $(x,y')$}
            \label{fig:blocking2}
    \end{subfigure}
\\
    \begin{subfigure}{0.31\textwidth}
    \centering
\begin{tikzpicture}[scale=0.4]
        \draw[black,very thick] (0,0) -- (10.5,0);
        \draw[very thick, dashed, ->] (1,0) arc (180:0:1.5);
        \draw[very thick, dashed, ->] (5,0) arc (180:0:2.5);
        \node at (1,-0.5) {\footnotesize~$x$};
        \node at (4,-0.5) {\footnotesize~$y$};
        \node at (5,-0.5) {\footnotesize~$x'$};
        \node at (10,-0.5) {\footnotesize~$y'$};
    \end{tikzpicture}
        \caption{blocked by $(x',y)$}
            \label{fig:blocking3}
    \end{subfigure}
    \begin{subfigure}{0.31\textwidth}
    \centering
    \begin{tikzpicture}[scale=0.4]
        \draw[black,very thick] (0,0) -- (8,0);
        \draw[very thick, dashed, ->] (2,0) arc (180:0:2);
        \draw[very thick, dashed, ->] (7,0) arc (0:180:3);
        \node at (1,-0.5) {\footnotesize~$y'$};
        \node at (2,-0.5) {\footnotesize~$x$};
        \node at (6,-0.5) {\footnotesize~$y$};
        \node at (7,-0.5) {\footnotesize~$x'$};
    \end{tikzpicture}
        \caption{blocked by $(x,y')$, $(x',y)$}
            \label{fig:blocking4}
    \end{subfigure}
    \caption{Forbidden patterns in stable matchings.}
    \label{fig:forbidden}
\end{figure}

\begin{lemma}\label{lm_stable_is_no_crossing}
    Any stable matching satisfies no-crossing.
\end{lemma}

\begin{proof}
Let $\pi$ be a stable matching. Towards a contradiction, suppose that $(x,y)$ and
$(x',y')$ are in the support of $\pi$, where $x\neq x'$, $y\neq y'$, and the
circles $O(x,y)$ and $O(x',y')$ intersect. There are, up to symmetry, two cases
to consider, depicted in Figures~\ref{fig:blocking1} and~\ref{fig:blocking2}.

First, $x<x'<y<y'$ in which case $\vert x'-y\vert<\vert x-y\vert$ and
$\vert x'-y\vert<\vert x'-y'\vert$, and so $\pi$ is blocked by $(x',y)$.
Second, we could have $x<y'<y<x'$. Now $\vert x-y'\vert < \vert x-y\vert$ and
$\vert x'-y\vert < \vert x'-y'\vert$. If $\pi$ were stable then neither
$(x',y)$ nor $(x,y')$ could block, so we would have
$\vert x-y \vert \leq \vert x'-y \vert$ and $\vert x'-y' \vert \leq \vert x-y' \vert$.
But then
\[
\vert x-y \vert \leq \vert x'-y \vert < \vert x'-y' \vert \leq \vert x-y' \vert
< \vert x-y \vert,
\]
a contradiction. Thus $\pi$ is not stable.
\end{proof}

We now describe two benchmark classes of markets for which the stable matching
can be obtained in a single step; in the general case, these will correspond to
the submarkets peeled off by the algorithm.

If~$\mu=\nu$, we call~$(\mu,\nu)$ a \emph{diagonal market}. In such markets, everyone can be matched with their ideal partner~$x=y$.
The opposite extreme is when nobody can be matched with their hypothetical ideal partner~$x=y$. Populations~$(\mu,\nu)$ form a \emph{disjoint market} if~$\mu$ and~$\nu$ are mutually singular, i.e., there is a measurable subset~$A\subset \R$ such that~$\mu(\R\setminus A)=0$ and~$\nu(A)=0$. An important special case corresponds to~$A=(-\infty, t]$ or~$A=[t,\infty)$ for some~$t\in \R$. In this case, the supports of~$\mu$ and~$\nu$ are ordered on the line, i.e., either~$\supp (\mu) \leq \supp (\nu)$ or~$\supp (\nu) \leq \supp (\mu)$. We call such~$(\mu,\nu)$ an \emph{anti-diagonal market}. Equivalently,~$(\mu,\nu)$ is an anti-diagonal market if~$\mu$ and~$\nu$ are mutually singular and 
$\rho$ changes sign only once.

\begin{lemma}\label{lm_benchmark_matching}
If~$(\mu,\nu)$ is  a diagonal or anti-diagonal market, then the stable matching~$\pi$ is unique and is assortative or anti-assortative, respectively. 
\end{lemma}
 In other words, for diagonal markets,~$\pi$ is supported on the solution~$F_\mu(x)=F_\nu(y)$, where~$F_\mu$ and~$F_\nu$ are the cumulative distribution functions of~$\mu$ and~$\nu$. Since~$\mu=\nu$,~$\pi$ is supported on the~$45$-degree line. In the anti-diagonal case,~$\pi$ is supported on the solution to~$F_\mu(x)=M-F_\nu(y)$, where $M=\mu(X)=\nu(X)$.  

\begin{proof}
For diagonal and anti-diagonal markets,~$\rho=\mu-\nu$ changes sign at most once. As shown by \cite*{mccann1999exact}, a matching satisfying the no-crossing condition is unique whenever the sign changes at most two times.
Assortative and anti-assortative matchings satisfy the no-crossing condition and, thus, are unique matchings with this property.
By Lemma~\ref{lm_stable_is_no_crossing}, they are unique stable ones.
\end{proof}   
A market~$(\mu',\nu')$ is a \emph{submarket} of~$(\mu,\nu)$ if~$\mu'\leq \mu$ and~$\nu'\leq \nu$. Note that 
$\mu'(X)=\nu'(Y)$ by the assumption that~$(\mu',\nu')$ is a market. For a submarket~$(\mu',\nu')$, the \emph{residual submarket} is defined by~$(\mu'',\nu'')=(\mu-\mu',\nu-\nu')$.

\begin{lemma}\label{lm_independent_diagonal_submarket}
For any market~$(\mu,\nu)$, there exists a unique diagonal submarket~$(\mu',\nu')$ such that the residual submarket~$(\mu'',\nu'')=(\mu-\mu',\nu-\nu')$ is disjoint.

A matching~$\pi$ is a stable matching of~$(\mu,\nu)$ if and only if~$\pi=\pi'+\pi''$, where~$\pi'$ is the (unique) stable matching of~$(\mu',\nu')$ and~$\pi''$ is a stable matching of~$(\mu'',\nu'')$.
\end{lemma}   
\begin{proof}
We first construct a diagonal submarket~$(\mu',\nu')$ such that the residual submarket~$(\mu'',\nu'')$ is disjoint. Let~$\rho=\mu-\nu$. We get that~$\rho$ can also be expressed as~$\mu''-\nu''$. 
By the Jordan decomposition theorem, there is a unique way to represent a signed measure~$\tau$ as~$\tau_+-\tau_-$, where~$\tau_\pm$ are positive measures that are mutually singular. We get that~$(\mu'',\nu'')$ must provide the Jordan decomposition of~$\rho$, i.e.,~$\mu''=\rho_+$ and~$\nu''=\rho_-$. Thus~$(\mu',\nu')$ is given by~$\mu'=\mu-\rho_+$ and~$\nu'=\nu-\rho_-$. It is a diagonal submarket since 
$$\mu'-\nu'=(\mu-\nu)-(\rho_+-\rho_-)=\rho-\rho=0.$$
Its uniqueness follows from the uniqueness of the Jordan decomposition.

Let~$\pi'$ be the stable matching of the diagonal submarket~$(\mu',\nu')$. By Lemma~\ref{lm_benchmark_matching}, matching~$\pi'$ is unique and is given by the assortative matching. If~$\pi''$ is a stable matching of~$(\mu'',\nu'')$, then~$\pi'+\pi''$ is a stable matching of~$(\mu,\nu)$. Indeed, by combining the two markets, we do not create any cross-market blocking pairs as any agent in~$\pi'$ is matched with their best partner.

It remains to show that any stable matching~$\pi$ of~$(\mu,\nu)$ can be represented as~$\pi'+\pi''$ with stable~$\pi'$ and~$\pi''$,
where~$\pi'$ is the stable matching of the diagonal submarket~$(\mu',\nu')$.  We adapt an argument from the optimal transport literature with metric costs \citep*[see][Proposition~2.9]{gangbo1996geometry}. 
Let~$\pi_\diag$ be the restriction of~$\pi$ to the diagonal~$\{(x,y)\colon x=y\}$, i.e.,~$\pi_\diag(A)=\pi(A\cap\{(x,y)\colon x=y\})$ for any measurable~$A$. Our goal is to show that~$\pi'=\pi_\diag$. Towards contradiction, assume this equality does not hold. Hence, the marginals~$\mu_\diag$ and~$\nu_\diag$ of~$\pi_\diag$ satisfy~$\mu_\diag\ne \mu'$ and~$\nu_\diag\ne \nu'$. By the uniqueness of the Jordan decomposition,~$\mu''=\mu-\mu_\diag$ and~$\nu''=\nu-\nu_\diag$ do not form a disjoint market, i.e., there is a set~$B\subset \R$ with~$\mu''(B)>0$ and~$\nu''(B)>0$. 
Consider a set~$S=\supp (\pi)\setminus \{(x,y)\colon x=y\}$. Let~$X_S$ and~$Y_S$ be the projections of the set~$S$. These are sets of full measure with respect to~$\mu''$ and~$\nu''$, and thus, the intersection~$X_S\cap Y_S$ is non-empty.  Pick~$t\in X_S\cap Y_S$. There are two couples~$(x_1,y_1),(x_2,y_2)\in\supp(\pi)$ with~$x_1=t$ and~$y_2=t$ and~$x_1\ne y_1$ and~$x_2\ne y_2$ which violate the no-crossing condition. By Lemma~\ref{lm_stable_is_no_crossing},~$\pi$ cannot be stable. This contradiction implies that~$\pi_\diag=\pi'$. We conclude that any stable~$\pi$ can be represented as~$\pi'+\pi''$, where~$\pi'$ is a stable matching of the diagonal submarket and~$\pi''$ is a stable matching of a disjoint submarket.
\end{proof}    

Lemma~\ref{lm_independent_diagonal_submarket} reduces the problem of constructing and showing the uniqueness of a stable matching for general~$(\mu,\nu)$ to these questions for disjoint markets.

A submarket~$(\mu',\nu')$ is \emph{independent} if 
members of~$(\mu',\nu')$ top-rank each other within~$(\mu,\nu)$, i.e., 
$|x-y'|>|x-y|~$ for~$\mu'$-almost all~$x$,~$\nu'$-almost all~$y$, and~$(\nu-\nu')$-almost all~$y'$ and
$|x'-y|>|x-y|$ for~$\mu'$-almost all~$x$,~$\nu'$-almost all~$y$, and~$(\mu-\mu')$-almost all~$x'$. 
Independent submarkets are important because they can be matched myopically, i.e., without worrying about the residual populations.

\begin{lemma}\label{lm_independent_submarkets}
    Let~$(\mu',\nu')$ be an independent submarket of~$(\mu,\nu)$ and let~$(\mu'',\nu'')=(\mu-\mu',\nu-\nu')$ be the residual submarket. A matching~$\pi$ is a stable matching of~$(\mu,\nu)$ if and only if~$\pi=\pi'+\pi''$, where~$\pi'$ is a stable matching of~$(\mu',\nu')$ and~$\pi''$ is a stable matching of~$(\mu'',\nu'')$.
\end{lemma}   
\begin{proof}
Let~$\pi'$ be a stable matching of an independent submarket~$(\mu',\nu')$ and~$\pi''$ be a stable matching of the residual submarket~$(\mu'',\nu'')$. Then~$\pi=\pi'+\pi''$ is a stable matching of~$(\mu,\nu)$ since combining the two markets cannot create cross-market blocking pairs by the independence of~$(\mu',\nu')$. 

Now, let~$\pi$ be a stable matching of~$(\mu,\nu)$. We show that~$\pi$ can be represented as~$\pi'+\pi''$, where~$\pi'$ is a stable matching of~$(\mu',\nu')$ and~$\pi''$ is a stable matching of~$(\mu'',\nu'')$. Note that the stability of a submarket follows from the stability of a market, and thus we only need to check that~$\pi=\pi'+\pi''$, where~$\pi'$ and~$\pi''$ are matchings.

Strict inequalities in the definition of independence imply that~$\mu'$ and~$\mu''$ cannot have any mass in common. In other words,~$\mu',\mu''$ and~$\nu',\nu''$ are mutually singular. By Hahn's theorem, there are disjoint sets~$X'$ and~$X''$ such that~$\mu'$ and~$\mu''$ are given by restricting~$\mu$ on~$X'$ and~$X''$, respectively. The disjoint sets~$Y'$ and~$Y''$ are constructed analogously. 

We now show that~$\pi$ cannot place positive mass on~$X'\times Y''$ and~$X''\times Y'$.
First, observe that~$\pi(X'\times Y'')=\pi(X''\times Y')$. Indeed,
$$\pi(X'\times Y'')-\pi(X''\times Y')=\pi(X'\times (Y'\cup Y''))-\pi((X'\cup X'')\times Y')=\mu(X')-\nu(Y')=0.$$
Second, we show that the common value is zero.
Towards contradiction, assume that~$\pi(X'\times Y'')=\pi(X''\times Y')>0$. Thus~$\supp(\pi)$ contains~$(x_1,y_2)\in X'\times Y''$ and 
$(x_2,y_1)\in X''\times Y'$. By independence of~$(\mu',\nu')$, the couple~$(x_1,y_1)$ is a blocking pair which contradicts stability of~$\pi$. We conclude that~$\pi$ has no mass outside of 
$X'\times Y'$ and~$X''\times Y''$. Consequently,~$\pi=\pi'+\pi''$, where~$\pi'$ is the restriction of~$\pi$ to~$X'\times Y'$ and~$\pi''$ is the restriction to~$X''\times Y''$. By the construction,~$\pi'$ is a matching of~$(\mu',\nu')$ and~$\pi''$ is a matching of~$(\mu'',\nu'')$, which are both stable by the stability of~$\pi$.
\end{proof}

We say that a matching $\pi$ is Monge if $\pi$ is supported on a graph of some function $f\colon \R\to\R$.
The following lemma is the key step in the proof of Proposition~\ref{th_algorithm_appendix}. Starting from disjoint~$(\mu,\nu)$ such that~$\rho=\mu-\nu$ changes sign~$K\geq 2$ times, this lemma allows us to reduce~$K$ sequentially until we reach an anti-diagonal market~($K=1$). 
\begin{lemma}\label{lm_disjoint_submarket}
Let~$(\mu,\nu)$
be a disjoint market such that~$\rho$ changes sign~$K\geq 2$ times. Then~$\mu$ and~$\nu$ can be represented as~$\mu=\mu_1'+\mu_2'+\mu''$ and~$\nu=\nu_1'+\nu_2'+\nu''$ so that 
\begin{enumerate}
\item\label{assert_diag}~$(\mu_1',\nu_1')$ and~$(\mu_2',\nu_2')$ are independent anti-diagonal submarkets of~$(\mu,\nu)$;
\item\label{assert_stable}  a matching~$\pi$ is a stable matching of~$(\mu,\nu)$ if and only if~$\pi=\pi_1'+\pi_2'+\pi''$, where~$\pi_i'$ is the (unique) stable matching of~$(\mu_i',\nu_i')$,~$i=1,2$, and~$\pi''$ is a stable matching of~$(\mu'',\nu'')$;
\item\label{assert_monge}   such a stable matching~$\pi$ is Monge if and only if~$\pi''$ is Monge; 
\item\label{assert_residual}   the residual submarket~$(\mu'',\nu'')$ is a disjoint market with~$\rho''=\mu''-\nu''$ changing sign at most~$K-1$ times;
\item \label{assert_algo} if~$\mu$ and~$\nu$ have piecewise constant density with~$m$ intervals of constancy, then~$(\mu_1',\nu_1')$ and~$(\mu_2',\nu_2')$ can be constructed in time of the order of~$m$.

\end{enumerate}
\end{lemma}    
\begin{proof}
    Since~$(\mu,\nu)$ is a non-atomic disjoint market with~$K\geq 2$, we can find points~$a_0<a_1<a_2<\ldots<a_{K}<a_{K+1}$ with~$a_0=-\infty$ and~$a_{K+1}=+\infty$ such that either~$\mu$ is supported on intervals~$I_k=(a_k,a_{k+1})$ with even~$k$ and~$\nu$ is supported on~$I_k$ with odd~$k$, or the other way around. 
    By the minimality of~$K$, each interval~$I_k$ carries strictly positive~$\mu$-mass  or strictly positive~$\nu$-mass. 

Consider~$\rho=\mu-\nu$.  We aim to pick a number~$\delta\geq 0$ and  a pair of points~$a_k^+\leq a_{k+1}^-$ in each closed interval~$\bar I_k=[a_k,a_{k+1}]$ so that the following conditions are satisfied:
\begin{enumerate}
    \item equal-weight condition:~$\rho([a_k^-,a_k^+])=0$ for~$k=1,\ldots, K$; 
    \item equal-distance condition:~$|a_k^+-a_k^-|=\delta$ for~$k=1,\ldots, K$.
\end{enumerate}  
The conditions on the collection of points~$a_k^\pm$ and~$\delta$ are closed since $\rho$ is non-atomic. Consider the maximal~$\delta\geq 0$ such that points satisfying the conditions exist and let~$a_k^\pm$ be the corresponding collection of points. 

As~$\delta$ cannot be increased further, there is an interval~$I_{k^*}=(a_{k^*},a_{k^*+1})$ with~$k^*=1,\ldots, K-1$ such that~$a_{k^*}^+=a_{k^*+1}^-$. In other words, the two points~$a_{k^*}^+$ and~$ a_{k^*+1}^-$ hit each other, which does not allow us to increase~$\delta$. 

Consider the two submarkets~$(\mu_1',\nu_1')$ and~$(\mu_2',\nu_2')$ cut from~$(\mu,\nu)$ by the intervals~$J_1^*=[a_{k^*}^-, a_{k^*}^+]$ and~$J_2^*=[a_{k^*+1}^-, a_{k^*+1}^+]$, i.e.,  
$$\mu_i'(A)=\mu(A\cap J_i^*),\qquad  \nu_i'(A)=\nu(A\cap J_i^*)\qquad i=1,2$$
for any measurable~$A\subset \R$. Denote by~$(\mu'',\nu'')$ the residual submarket~$(\mu'',\nu'')=(\mu-\mu_1'-\mu_2',\nu-\nu_1'-\nu_2')$.

The submarket~$(\mu_1',\nu_1')$ is an independent anti-diagonal submarket of~$(\mu,\nu)$. To show this,  assume without loss of generality that the interval
$I_{k^*}$ carries a positive~$\mu$-weight. Hence,~$\mu_1'$ is supported on~$[a_{k^*},a_{k^*}^+]$ and~$\nu_1'$ on~$[a_{k^*}^-,a_{k^*}]$. By the equal-weight condition,~$\mu_1'(\R)=\nu_1'(\R)$ and so~$(\mu_1',\nu_1')$ is an anti-diagonal submarket of~$(\mu,\nu)$. The equal-distance condition implies that any~$x\in [a_{k^*},a_{k^*}^+]$ and~$y\in [a_{k^*}^-,a_{k^*}]$ are within~$\delta$ from each other. On the other hand, by the maximality of $\delta$, the residual points are at least $\delta$ away; that is, the distances between~$x$ and~$y' \in \supp (\nu'')$ and between~$x'\in \supp (\mu'')$ and~$y$ are at least~$\delta$.
Hence,~$(\mu_1',\nu_1')$ is independent. Similarly,~$(\mu_2',\nu_2')$ is an independent anti-diagonal submarket of~$(\mu,\nu)$. We get assertion~\ref{assert_diag}.

Applying Lemma~\ref{lm_independent_submarkets}, 
we conclude that~$\pi$ is a stable matching of~$(\mu,\nu)$ if and only if~$\pi=\pi_1'+\pi_2'+\pi''$, where~$\pi_i'$ is the unique stable matching of~$(\mu_i',\nu_i')$,~$i=1,2$, and~$\pi''$ is a stable matching of~$(\mu'',\nu'')$. We obtain assertion~\ref{assert_stable}. 

The stable matching for anti-diagonal markets is Monge (Lemma~\ref{lm_benchmark_matching}). Since~$\mu'_1,\mu'_2$, and~$\mu''$ are mutually singular (and similarly for~$\nu$), we conclude that a stable matching~$\pi$ is Monge if and only if~$\pi''$ is Monge. Assertion~\ref{assert_monge} is proved.

The residual submarket~$(\mu'',\nu'')$ is a disjoint market, and~$\rho''=\mu''-\nu''$ changes sign at most~$K-1$ times. Indeed, the interval~$I_{k^*}$ is removed from both~$\mu$ and~$\nu$, and so the number of imbalance regions is reduced by at least~$1$. Thus assertion~\ref{assert_residual} holds.

Finally, suppose that~$\mu$ and~$\nu$ have piecewise constant densities with at most~$m$ intervals of constancy. Without loss of generality, these intervals are common and consecutive, i.e., there is a collection of points~$b_0<b_1<\ldots< b_m$ such that the densities of~$\mu$ and~$\nu$ on an interval~$J_i=(b_i,b_{i+1})$ are constants~$f_i$ and~$g_i$, respectively. 
We can also assume that~$(b_i)_{i=0,\ldots,m}$ is a subsequence of the sequence~$(a_k)_{k=0,\ldots, K+1}$, which defines intervals~$I_k$.

To construct the submarket~$(\mu_1',\nu_1')$ and~$(\mu_2',\nu_2')$, we need to identify an interval~$I_{k^*}$ defined above. This can be done as follows. For each three consecutive intervals~$I_{k-1},I_k,I_{k+1}$, we look for a triplet of points~$\alpha_k^-\in \bar{I_{k-1}}$,~$\alpha_k\in \bar{I_k}$ and~$\alpha_{k}^+\in \bar{I_{k+1}}$ such that 
\begin{equation}\label{eq_I_k_test}
\rho([\alpha_k^-,\alpha_k])=\rho([\alpha_k,\alpha_k^+])=0\qquad \text{and}
\qquad |\alpha_k-\alpha_k^-|=|\alpha_k-\alpha_k^+|.
\end{equation}
If~$I_{k-1},\ I_k,\ I_{k+1}$ contain~$m_k$ intervals of constancy~$J_i$, then finding a solution~$(\alpha_k^-,\alpha_k, \alpha_k^+)$ or checking that no such solution exists can be done in~$O(m_k)$ operations. 
Indeed, to solve the system~\eqref{eq_I_k_test}, it is enough to test each of the subintervals of constancy for whether it contains a solution.
To see this, note  that the cumulative distribution function~$R_k(x)=\rho([a_{k-1},x])$ for~$x\in I_{k-1}\cup I_k\cup I_{k+1}$ is piecewise linear with~$m_k$ intervals of linearity.  Testing corresponds to solving a linear system with three unknowns. We conclude that solving~\eqref{eq_I_k_test} boils down to 
solving of the order of~$m_k$ linear systems of given size and thus requires~$O(m_k)$ operations. 
After finding a solution for each triplet of consecutive intervals, we pick~$k=k^*$ that minimizes~$\delta_k=|\alpha_k-\alpha_k^-|$. The interval~$I_{k^*}$  and points~$(a_{k^*}^-,a_{k^*}^+,a_{k^*+1}^-,a_{k^*+1}^+)=(\alpha_{k^*}^-, \alpha_{k^*}, \alpha_{k^*}, \alpha_{k^*}^+)$ determine~$(\mu_1',\nu_1')$ and~$(\mu_2',\nu_2')$. 

We need~$O(m_k)$ operations per each interval~$I_k$, and thus, the total number of operations is of the order~$\sum_k m_k$, i.e., of the order of the number of intervals of constancy~$m$. Assertion~\ref{assert_algo} is proved.
\end{proof}

We are now ready to prove Proposition~\ref{th_algorithm_appendix}.

\paragraph{Proof of Proposition~\ref{th_algorithm_appendix}.}

Consider~$(\mu,\nu)$ and assume that~$\rho=\mu-\nu$ changes sign~$K$ times. 
By Lemma~\ref{lm_independent_diagonal_submarket}, any stable matching~$\pi=\pi_\diag+\pi_1$, where~$\pi_\diag$ is a unique stable matching of a diagonal market~$(\mu-\rho_+, \nu-\rho_-)$ and~$\pi_1$ is a stable matching of the residual submarket~$(\mu_1,\nu_1)$ with~$\mu_1=\rho_+$ and~$\nu_1=\rho_-$.

The residual submarket~$(\mu_1,\nu_1)$ is disjoint and~$\rho_1=\mu_1-\nu_1$ changes sign~$K$ times since~$\rho_1=\rho$.  We construct a sequence of submarkets~$(\mu_k,\nu_k)$ inductively starting from~$(\mu_1,\nu_1)$. Suppose~$(\mu_k,\nu_k)$ is already constructed and let~$(\mu_k',\nu_k')$ be its independent submarket from Lemma~\ref{lm_disjoint_submarket}. We then define 
$(\mu_{k+1},\nu_{k+1})$ by 
$$\mu_{k+1}=\mu_k-\mu_k'\quad \text{and}\quad \nu_{k+1}=\nu_k-\nu_k'.$$
Let~$\pi_k'$ be a stable matching of~$(\mu_k',\nu_k')$ which is unique since~$\rho_k'=\mu_k'-\nu_k'$ changes sign at most two times. Thus any stable matching of~$(\mu,\nu)$ can be represented as
$$\pi=\pi_\diag+\pi_1'+\pi_2'+\ldots+\pi_{k-1}'+\pi_{k},$$
where~$\pi_{k}$ is a stable matching of~$(\mu_{k},\nu_{k})$. By Lemma~\ref{lm_disjoint_submarket},~$\rho_{k}$ changes sign at most~$K-k+1$ times, so there is~$L\leq K$ such that~$\rho_{L}$ changes sign at most one time.
Abusing the notation, denote the stable matching of~$(\mu_{L},\nu_{L})$ by 
$\pi_{L}'$. This stable matching is unique by Lemma~\ref{lm_benchmark_matching}.

We conclude that a stable matching~$\pi$  of~$(\mu,\nu)$ is unique and has the following form
$$\pi=\pi_{\diag}+\underbrace{\pi_1'+\pi_2'+\ldots+\pi_{L-1}'+\pi_L'}_{\mbox{$L\leq K$ terms}}.$$
By Lemma~\ref{lm_disjoint_submarket}, $\pi_1'+\pi_2'+\ldots+\pi_{L-1}'+\pi_L'$ is a Monge matching. Thus $\pi$ is a combination of a diagonal matching and a Monge matching. Consequently, for each $x$ there are at most two distinct $y$ such that $(x,y)$ is in the support of $\pi$. Moreover, if there are two such~$y$, one of them necessarily equals~$x$. 
\smallskip

We now consider the computational complexity of constructing~$\pi$ for~$\mu$ and~$\nu$ having piecewise constant densities with at most~$m$ intervals of constancy. Without loss of generality, these intervals are common and consecutive, i.e., there is a collection of points~$b_0<b_1<\ldots< b_m$ such that the densities of~$\mu$ and~$\nu$ on an interval~$J_i=(b_i,b_{i+1})$ are constants~$f_i$ and~$g_i$, respectively. The collection of all these numbers is the input of the algorithm.

Constructing the diagonal~$\pi_\diag$ requires a linear number of operations in~$m$. Indeed,~$\mu-\rho_+=\nu-\rho_-$ have density~$\min\{f_i,g_i\}$ on~$J_i$ and~$\pi_{\diag}$ has the corresponding density on the diagonal of~$J_i\times J_i$.

The complexity bottleneck corresponds to finding independent submarkets from Lemma~\ref{lm_disjoint_submarket}. By Assertion~\ref{assert_algo}, this requires of the order of~$m$ operations for each~$(\mu_k,\nu_k)$. The number of steps~$k$ is bounded by the number of times~$\rho$ changes its sign, and this number cannot exceed~$m-1$. We conclude that the stable matching can be constructed in~$O(m^2)$ operations.

\section{Proof of Corollary~\ref{cor_existence}}\label{app_corr_existence}
We first show that the sets~$\Pi_{+\infty}^{U}(\mu,\nu)$ and ~$\Pi_{-\infty}^{U}(\mu,\nu)$, corresponding to~$\alpha\to \pm\infty$, are non-empty. The argument for the two cases is identical, and so we focus on~$\Pi_{+\infty}^{U}(\mu,\nu)$. Consider a sequence~$\alpha_n\to+\infty$ and let~$\pi_n$ be a solution to the optimal transportation problem from Theorem~\ref{th_equivalence} with~$\alpha=\alpha_n$. Such a solution is guaranteed to exist under our assumptions on~$U$; see the discussion after the theorem. The set of all transportation plans~$\Pi(\mu,\nu)$ for~$\mu\in \Delta(X)$ and~$\nu\in\Delta(Y)$ with Polish~$X$ and~$Y$ is sequentially compact in the topology of weak convergence; this is a corollary of Prokhorov's theorem, see Lemma~4.4 in \cite*{villani2009optimal}. Thus, possibly passing to a subsequence, we conclude that the sequence~$\pi_n$ converges weakly to some~$\pi_{+\infty}\in \Pi(\mu,\nu)$. Thus~$\Pi_{+\infty}^{U}(\mu,\nu)$ and~$\Pi_{-\infty}^{U}(\mu,\nu)$ are non-empty.

We now show that~$\Pi_{+\infty}^{U}(\mu,\nu)$ consists of stable matchings. In other words, we demonstrate that~$\pi_{+\infty}$ is stable. Consider a continuous function~$f\colon (X\times Y)^2\to \R$ given by
$$f(x_1,y_1,x_2,y_2)=\max\Big\{0, \ \ U(x_1,y_2)-\max\left\{U(x_1,y_1),\, U(x_2,y_2)\right\}\Big\}.$$
A matching~$\pi$ is stable if and only if 
$\int f \diff(\pi\times \pi)=0$. Note that for an~$\varepsilon$-stable matching, this integral does not exceed~$\varepsilon\cdot \mu(X)\cdot \nu(Y)$. We obtain
$$\int f \diff(\pi_{+\infty}\times \pi_{+\infty})=\lim_{n\to\infty}\int f \diff(\pi_n\times \pi_n) \leq \lim_{n\to\infty} \frac{\ln 2}{\alpha_n}\cdot \mu(X)\cdot \nu(Y)=0.$$
Since the left-hand side is non-negative, we get that it equals zero. Thus 
$\pi_{+\infty}$ is stable. We conclude that~$\Pi_{+\infty}^{U}(\mu,\nu)$ consists of stable matchings, and so a stable matching exists.

To show that~$\Pi_{-\infty}^{U}(\mu,\nu)$ consists of egalitarian matchings, consider a weak limit~$\pi_{-\infty}$ of a sequence of matchings~$\pi_{\alpha_n}$, where~$\alpha_n\to -\infty$. Fix~$\varepsilon>0$ and  let~$C_\varepsilon$ be the set of hypothetical couples whose utility is below the egalitarian lower bound by more than~$\varepsilon$
$$C_\varepsilon=\left\{(x,y)\in X\times Y\colon U(x,y)< \mathcal{U}_{\min}^*(\mu,\nu)-\varepsilon\right\}.$$
By the continuity of~$U$, the set~$C_\varepsilon$ is open and thus
$\pi_{-\infty}(C_\varepsilon)\leq \liminf_{n\to\infty} \pi_{\alpha_n}(C_\varepsilon)$.
By Theorem~\ref{th_equivalence}, the right-hand side goes to zero, and thus~$\pi_{-\infty}(C_\varepsilon)=0$ for any~$\varepsilon>0$. Sets~$C_\varepsilon$  are decreasing in~$\varepsilon$. Hence,
$$\pi_{-\infty}(C_0)= \pi_{-\infty}\left(\cup_{\varepsilon>0} C_\varepsilon \right)=\lim_{\varepsilon\to 0} \pi_{-\infty}(C_\varepsilon)=0.$$
Since~$\pi_{-\infty}(C_0)=0$, for~$\pi_{-\infty}$-almost all couples~$(x,y)$ the utility~$U(x,y)$ is at least~$\mathcal{U}_{\min}^*(\mu,\nu)$, i.e.,~$\pi_{-\infty}$ is an egalitarian matching. Therefore, all elements of~$\Pi_{-\infty}^{U}(\mu,\nu)$ are egalitarian matchings.

Finally, 
the weak closedness of~$\Pi_{+\infty}^{U}(\mu,\nu)$ and~$\Pi_{-\infty}^{U}(\mu,\nu)$ follows from the weak closedness of the set of solutions via the standard diagonal procedure.

\section{Multi-sided matching}\label{sec_multi_side}
Markets with more than two sides, referred to as multi-partner matching markets, are known to be challenging. Stable matchings may not exist, even in small finite settings. It is therefore notable that our results for markets with aligned preferences extend to this setting, opening up the theory to additional applications such as teams, clubs, and coalition formation.

First, we revisit the standard finite model. Let $X_1,\ldots,X_k$ be finite sets with the same cardinality, representing $k\geq 2$ sides of the market. For each side $i\in\{1,\ldots,k\}$ of the market and $x_i\in X_i$, agent $x_i$ has a complete and
transitive preference $\succeq_{x_i}$ over partner tuples
$x_{-i}\in\prod_{j\neq i}X_j$. A {matching} is a collection of $k$-tuples
$(x_1,\ldots,x_k)$ that partitions $X_1\cup\cdots\cup X_k$ so that every
$x_i\in X_i$ appears in exactly one tuple; write $M(x_i)$ for the tuple
containing $x_i$. A $k$-tuple $x=(x_1,\ldots,x_k)$ {blocks} a matching
$M$ if each participant strictly prefers $x$ to her assigned tuple,
$x_{-i}\succ_{x_i}M(x_i)_{-i}$ for all $i$; a matching is {stable} if there are no
blocking $k$-tuples.

The multi-partner matching problem is well known to be intractable. 
Even with three sides ($k = 3$) and preferences that are additively separable---so that each agent’s preference $\succ_{x_i}$ can be represented by an additively separable utility function $u_i(x) = \sum_{j \neq i} v_i(x_i, x_j)$---a stable matching may fail to exist \citep*{alkan1988nonexistence}. Moreover, determining whether any stable matching exists is an NP-hard problem \citep*{ng1991three}.
The case of aligned preferences turns out to be an exception: a stable matching is guaranteed to exist, and the theory we have developed for two-sided markets readily extends to multi-partner matching.

To parallel the two-sided model in Section~\ref{sec_epsiloin_stable_general}, we
move to a general measure-theoretic formulation. For each $i$, let $X_i$ be a
complete separable metric space with its Borel $\sigma$-algebra, and let
$\mu_i\in \Delta(X_i)$ represent the distribution of types on side $i$. Denote $X=X_1\times \ldots \times X_k$.
The market is balanced
since $\mu_1(X_1)=\cdots=\mu_k(X_k)=1$ but unbalanced markets can be
handled by adding “dummy’’ agents. Each $x_i\in X_i$ has a (complete,
transitive) preference $\succeq_{x_i}$ over $X_{-i}:=\prod_{j\neq i}X_j$. The
preferences are \textbf{aligned} if there exists
$U\colon X\to\R$ 
such that, for every $i$, every $x_i\in X_i$, and all $x_{-i},x'_{-i}\in X_{-i}$,
\begin{equation}\label{eq_potential_k}
U(x_i,x_{-i}) \;\ge\; U(x_i,x'_{-i})
\quad\Longleftrightarrow\quad
x_{-i}\succeq_{x_i} x'_{-i}.
\end{equation}
As in the two-sided case, once alignment holds we treat $U$ as primitive and
refer to it as the (common) utility or potential.

A~\textbf{matching}~$\pi$ is a probability distribution on~$X$
with marginals~$\mu_i$ on~$X_i$; the
set of all matchings is denoted by $\Pi(\mu_1,\ldots, \mu_k)$. To express stability in this setting, draw $k$ existing tuples independently:
$\pi^{\times k}$ denotes the product measure on $X^k$, with realizations
$x^1,\ldots,x^k\in X$. Consider the ``recombined'' tuple
$(x^1_1,\ldots,x^k_k)$. Stability requires that for $\pi^{\times k}$-almost all
such collections, at least one participant weakly prefers her current tuple to
the recombined one. Under aligned preferences this is
equivalent to the potential inequality
\begin{equation}\label{eq_stable_common_k}
U(x^1_1,\ldots,x^k_k)\;\le\;\max_{i=1,\ldots,k} U(x^i).
\end{equation}
We call a matching $\pi$ \textbf{stable} if \eqref{eq_stable_common_k} holds for
$\pi^{\times k}$-almost all $(x^1,\ldots,x^k)$. For $\varepsilon\ge0$, $\pi$ is
$\varepsilon$-\textbf{stable} if we allow an $\varepsilon$-slack in
\eqref{eq_stable_common_k}:
\begin{equation*}
U(x^1_1,\ldots,x^k_k)\;\le\;\max_{i=1,\ldots,k} U(x^i)+\varepsilon.
\end{equation*}
When $U$ is continuous, Lemma~\ref{lm_deterministic_equivalence} (extended
verbatim to the $k$-sided case) allows an equivalent pointwise formulation by
requiring these inequalities for all $x^i\in\supp(\pi)$.
The notions of $\varepsilon$-egalitarian matchings and optimal welfare $W^*$
extend to the $k$-sided case in the same way as for two-sided markets.

As in the two-sided case, we consider a CARA-based version of Atkinson’s inequality index
\begin{equation*}
\text{Atkinson}_{\alpha}(\pi)=\int_{X}\frac{1-\exp(\alpha\cdot U(x))}{\alpha}\diff\pi(x),
\end{equation*}
and assume that a social planner seeks to minimize this index across all matchings:
\begin{equation}\label{eq_Atkinson_minimization_k_sided}
    \min_{\pi\in \Pi(\mu_1,\ldots, \mu_k)}\text{Atkinson}_\alpha(\pi).
\end{equation}
This constitutes a particular $k$-marginal optimal transportation problem.\footnote{A general $k$-marginal optimal transportation problem is defined as follows: given a cost function $c\colon X\to\R$, find a matching $\pi$ that minimizes
$
\min_{\pi\in \Pi(\mu_1,\ldots, \mu_k)}\int_{X} c(x)\, \diff\pi(x).$} 
Using this connection to optimal transport, we now extend Theorems~\ref{th_equivalence} and \ref{th_welfare} to the~$k$-partner setting.

\begin{theorem}\label{th_k_sided}
       Consider a~$k$-partner matching market with aligned preferences and bounded continuous potential~$U\colon X_1\times\ldots\times X_k\to \R_+$. Let~$\pi^*$ be a solution to~\eqref{eq_Atkinson_minimization_k_sided}. The following assertions hold:
       \begin{enumerate}
       \item \label{asert_stable_k_side} If~$\alpha>0$, then~$\pi^*$ is~$\varepsilon$-stable with~$\varepsilon=\frac{\ln k}{\alpha}$.
       \item \label{asert_egalitarian_k_side} If~$\alpha<0$, then~$\pi^*$ 
is~$\varepsilon$-egalitarian with~$\varepsilon=\frac{\max\{1,\, \ln |\alpha|\}}{|\alpha|}.$
\item \label{asert_welfare_k_side} 
Any~$\varepsilon$-stable matching~$\pi$ satisfies
$W(\pi)\geq \frac{1}{k}\left(W^*-\varepsilon\right),$ and
moreover,~$\pi$ is~$\varepsilon'$-egalitarian with~$\varepsilon'=\max\left\{\frac{k-1}{k},\, \varepsilon\right\}.$
\end{enumerate}    
\end{theorem}
As we see, the factor~$2$ in Theorems~\ref{th_equivalence} and~\ref{th_welfare} corresponds to the number of sides of the market. Thus, $2$ is replaced by~$k$ in the first and the third claims of Theorem~\ref{th_k_sided}. Interestingly, the approximation quality for the egalitarian goal is not sensitive to the number of sides $k$.

Before proving the theorem, we outline an example illustrating how it can be applied beyond the standard settings of team, club, or coalition formation.  The following example concerns organ exchanges under medical compatibility constraints.
It illustrates how such constraints can be encoded by modeling the market as
multi-sided with a common potential. The formulation is intentionally stylized
and serves only to highlight what can be achieved in $k$-sided matching
environments \citep*{roth2005pairwise,roth2007efficient}.
\begin{example}[Organ exchanges]\label{ex_organ}
An agent is a patient-donor pair, where the donor is willing to donate an organ to the patient, but the donation is not feasible because the two are biologically incompatible. For practical reasons---see \citep*{roth2005pairwise,roth2007efficient}---we rule out complicated organ trades and focus on pairwise exchanges. At first glance this looks like a one-sided (roommate)
problem in which any pair could trade with any other, but compatibility
constraints can be captured cleanly by recasting the environment as a
multi-sided market with aligned preferences.

Suppose there are $k$ {types} of agents (e.g., blood-type
profiles such as $(\mathrm O,\mathrm A)$, $(\mathrm A,\mathrm B)$, etc.; other
markers can be incorporated similarly). For each type $i\in\{1,\ldots,k\}$,
let each pair be described by an organ-size vector $s^i=(s^i_1,s^i_2)\in\R_+^2$
(patient size, donor size).  Let $\mu_i$ be the distribution on $\R_+^2$ of 
size vectors within type $i$, normalized so that $\mu_i(\R_+^2)$ is the total
share of type $i$ agents. Let
$G=(\{1,\ldots,k\},E)$ be a {compatibility graph}: $(i,j)\in E$ if and only if a
donor from type $i$ is medically compatible with a patient from type $j$
(and vice versa). For every compatible pair $(i,j)\in E$ and size vectors
$s^i,s^j\in\R_+^2$, let $q_{i,j}(s^i,s^j)$ denote the match quality of a
two-way exchange between types $i$ and $j$.

To embed this as a $k$-sided market, treat each type $i$ as a separate side. Add a dummy point $d_i$ on side $i$ and assign it mass $1-\mu_i(\R_+^2)$, so that each side has total mass one.
 Formally, set $X_i=\R_+^2\cup\{d_i\}$ and redefine $\mu_i$ on $X_i$ by placing any extra mass at $d_i$. We then define a potential $U:\prod_{i=1}^k X_i\to\R$ by
$$
U(x_1,\ldots,x_k)=
\begin{cases}
q_{i,j}(x_i,x_j), & \mbox{if~$(i,j)\in E$, and~$x_i, x_j$ are the only non-dummy types}\\
0, & \mbox{if there is at most one non-dummy type}\\
-M, & \mbox{otherwise}
\end{cases},
$$
where~$M>0$ is a large number.
Thus $U$ assigns quality $q_{i,j}$ to feasible
two-way exchanges, gives $0$ to configurations in which a non-dummy agent
remains unmatched, and penalizes any configuration that is infeasible or
attempts a multi-way exchange.

This construction transforms the organ exchange problem into a $k$-sided matching market with a common match-quality potential. The general results established above then apply: one can compute stable or $\varepsilon$-stable matchings and analyze their egalitarian and welfare properties, as well as the associated trade-offs, within the same multi-marginal optimal-transport framework.

\end{example}
\begin{proof}[\textbf{\emph{Proof of Theorem~\ref{th_k_sided}}}]\label{app_proof_k_sided}
We extend the proof techniques used in the two-sided case. 

\noindent\textit{Proof of Assertion~\ref{asert_stable_k_side}.} 
     \cite*{kim2014general} generalize the two-marginal~$c$-cyclic monotonicity condition that we used in the proof of Theorem~\ref{th_equivalence}. They consider a~$k$-marginal optimal transportation problem with a bounded continuous cost~$c\colon X_1\times \ldots\times X_k\to \R$ and show that~$\pi$ is optimal if and only if~$\supp(\pi)$ is~$c$-cyclically monotone. Here, a set~$T\subset X_1\times \ldots\times X_k$ is~$c$-cyclically monotone if for any collection of~$k$-tuples 
~$(x_1^{i},\ldots, x_k^i)\in T$,~$i=1,\ldots, n$, and a collection of~$k$ permutations~$\sigma_1,\ldots, \sigma_k$ of the set~$\{1,\ldots, n\}$, the following inequality holds:
$$\sum_{i=1}^n c(x_1^i,\ldots, x_k^i)\leq \sum_{i=1}^n c(x_1^{\sigma_1(i)},\ldots, x_k^{\sigma_k(i)}).$$

      We now apply this result to~$c=c_\alpha$ with
$$ c_\alpha(x)=\frac{1-\exp(\alpha\cdot U(x))}{\alpha}.$$
Let $\pi^*$ be a solution to the optimal transportation problem with this cost.  
      Take the number of points~$n$ equal to the number of marginals~$k$ and put~$\sigma_i(1)=i$. We obtain that for  any collection~$x^i =(x_1^{i},\ldots, x_k^i)\in \supp(\pi^*)$,~$i=1,\ldots, k$
      \begin{align*}
        \sum_{i=1}^k \exp\left(\alpha \cdot U(x^i)\right) & \geq \sum_{i=1}^k \exp\left(\alpha \cdot U(x_1^{\sigma_1(i)},\ldots, x_k^{\sigma_k(i)})\right) \\ & = \exp\left(\alpha \cdot U(x_1^1,\ldots, x_k^k)\right) + \sum_{i=2}^k \exp\left(\alpha \cdot U(x_1^{\sigma_1(i)},\ldots, x_k^{\sigma_k(i)})\right).  
      \end{align*}
      Thus 
$$k\cdot \max_i \exp\left(\alpha \cdot U(x^i)\right) \geq \exp\left(\alpha \cdot U(x_1^1,\ldots, x_k^k)\right).$$
      Taking logarithm, we get 
$$\max_i U(x^i) +\frac{\ln k}{\alpha} \geq  U(x_1^1,\ldots, x_k^k)$$
and conclude that $\pi^*$ is~$\varepsilon$-stable with~$\varepsilon=\frac{\ln k}{\alpha}$. 
\medskip

\noindent\textit{Proof of Assertion~\ref{asert_egalitarian_k_side}.} 
Let $\mathcal{U}_{\min}^*(\mu_1,\ldots, \mu_k)$ be the egalitarian lower bound and $\pi'$ be a matching such that $U(x)\geq \mathcal{U}_{\min}^*-\delta$ where $\delta>0$ is small.
Consider a set
$$C=\left\{x\in X \colon U(x)< \mathcal{U}_{\min}^*-\varepsilon \right\}.$$
For $\alpha<0$, a solution~$\pi^*$ to the optimal transportation problem with cost~$c_\alpha$ maximizes
$\int -\exp({\al\cdot U(x)}) \diff \pi^*$ over all matchings and thus
$$\int -\exp({\al\cdot U(x)}) \diff \pi^*(x)\geq \int -\exp({\al\cdot U(x)}) \diff \pi'(x)\geq -\exp\big({\al\cdot (\mathcal{U}_{\min}^*-\delta)}\big).$$
Since $\delta$ is arbitrary, we obtain
\begin{align*}
-\exp\big({\alpha \cdot \mathcal{U}_{\min}^*}\big) &  \leq \int -\exp({\al\cdot U(x)}) \diff \pi^*(x) \\
& = \int_{C} -\exp({\al\cdot U(x)}) \diff \pi^*(x)+ \int_{X\setminus C} -\exp({\al\cdot U(x)}) \diff \pi^*(x)\\
& \leq \int_{C} -\exp({\al\cdot U(x)}) \diff \pi^*(x) \\
& \leq -\exp\big({\alpha \cdot (\mathcal{U}_{\min}^*-\ep)}\big)\cdot \pi^*(C).
\end{align*}
Thus
$$\al\cdot  \mathcal{U}_{\min}^*\geq \al\cdot  (\mathcal{U}_{\min}^*-\ep) + \ln (\pi^*(C)),$$
and so~$\al\cdot  \ep\geq \ln(\pi^*(C)).$ Plugging in~$\varepsilon=\frac{\max\{1,\, \ln |\alpha|\}}{|\alpha|}$, we get 
$$\pi^*(C)\leq \min \left\{\exp(-1),\, \frac{1}{|\alpha|}\right\}$$
and thus 
$\pi^*(C)\leq \varepsilon.$
We conclude that~$\pi^*$ is~$\varepsilon$-egalitarian with~$\varepsilon=\frac{\max\{1,\, \ln |\alpha|\}}{|\alpha|}$.

\medskip

\noindent\textit{Proof of Assertion~\ref{asert_welfare_k_side}.} 
Let~$\pi$ be an~$\varepsilon$-stable matching. Since~$U$ is continuous,~$\varepsilon$-stability implies that, for all collections of points~$x^i=(x_1^i,\ldots,x_k^i)\in \supp(\pi)$,~$i=1,\ldots, k$,
\[
U(x^1_1,\ldots,x^k_k) \leq \max_{i=1,\ldots,k}U(x^i)+\varepsilon \leq U(x^1)+\ldots+ U(x^k)+ \varepsilon.
\]
Let~$\pi'$ be any other matching.
Denote $X=X_1\times \ldots\times X_k$ and 
consider a distribution
$$\lambda\in \Delta\big(\underbrace{X\times\ldots \times X}_{\text{$k$ times}}\big)$$ such that
the marginals of~$\lambda$ on each~$X$ are equal to~$\pi$, and the marginal on the set~$\{(x_1^1,\ldots, x_k^k):(x^1,\ldots,x^k)\in X\times \ldots \times X\}$ coincides with~$\pi'$. As in the case of two-sided markets, such a $\lambda$ can be obtained by the standard gluing construction, using regular conditional distributions of $\pi$ and the matching $\pi'$. We  get
\begin{align*}
    W(\pi')&=\int_{X} U(x_1^1,\ldots,x_k^k)\diff\pi'(x_1^1,\ldots,x_k^k)=\int_{X\times \ldots \times X} U(x_1^1,\ldots,x_k^k)\diff\lambda(x^1,\ldots,x^k)\\
    &\leq \int_{X\times \ldots \times X} \left(U(x^1)+\ldots+ U(x^k)+ \varepsilon\right)\diff\lambda(x^1,\ldots,x^k)\\
    &=k\cdot W(\pi)+\varepsilon. 
    \end{align*}
    We thus obtain~$W(\pi)\geq \frac{1}{k}\left(W(\pi')-\varepsilon\right)$
     for any matching~$\pi'$. In particular, this inequality holds for~$\pi'$ maximizing welfare. Thus~$W(\pi)\geq \frac{1}{k}\left(W^*-\varepsilon\right)$.

     We now show that a substantial fraction of agents in an~$\varepsilon$-stable matching~$\pi$ have utilities above the egalitarian utility level~$\mathcal{U}_{\min}^*$. Consider
$$C_\delta=\left\{x\in X\colon U(x)< \mathcal{U}_{\min}^*-\varepsilon-\delta \right\},$$
where $\delta$ is a small positive number.
     Our goal is to bound ~$\pi(C_0)$, but we will bound $\pi(C_\delta)$ with $\delta>0$ first. Fix $\delta>0$. 
     As above, let~$\lambda$ be a distribution on~$X\times\ldots  \times X$  with marginals~$\pi$ on each copy of~$X$ and~$\pi'$ on~$(x_1^1,\ldots,x_k^k)$. Let~$\pi'$ in this construction be a matching such that $U(x)\geq \mathcal{U}_{\min}^*-\delta$ for $\pi'$-almost all~$x$.
     Thus~$U(x_1^1,\ldots, x_k^k)\geq \mathcal{U}_{\min}^*-\delta$ on a set of full~$\lambda$-measure.
     By~$\varepsilon$-stability,~$U(x_1^1,\ldots, x_k^k)\leq \max_i U(x^i)+ \varepsilon$ and thus
$$\max_i U(x^i)+ \varepsilon \geq \mathcal{U}_{\min}^*-\delta$$
on a set of full~$\lambda$-measure. However, on $C_\delta\times\ldots \times C_\delta$, the inequality is reversed
$$\max_i U(x^i)+ \varepsilon < \mathcal{U}_{\min}^*-\delta$$
and thus this set must have zero measure 
$$\lambda(C_\delta\times\ldots\times C_\delta)=0.$$
On the other hand, the measure of this set admits the following lower bound
     \begin{align*}
     \lambda(C_\delta\times\ldots\times C_\delta)&\geq 1 - \sum_i\lambda(\{(x^1,\ldots, x^k)\colon x^i\not\in C_\delta  \})\\
     &= 1- \sum_i (1-\pi(C_\delta))\\
     &=k\cdot \pi(C_\delta)-(k-1).
     \end{align*}
     Thus $k\cdot \pi(C_\delta)-(k-1)\leq 0$ or, equivalently, $\pi(C_\delta)\leq \frac{k-1}{k}$. The sequence of sets $C_\delta$ is decreasing in $\delta$ and $C_0=\cup_{\delta>0} C_\delta$. Hence,
$$\pi(C_0)=\lim_{\delta\to 0}\pi(C_\delta)\leq \frac{k-1}{k}$$
and thus any~$\varepsilon$-stable matching~$\pi$ is~$\varepsilon'$-egalitarian with~$\varepsilon'=\max\{\frac{k-1}{k}, \varepsilon\}$.
\end{proof}
\end{document}